\documentclass[reprint,amsfonts,amssymb,amsmath,showkeys,prb,twocolumn,longbibliography,nofootinbib]{revtex4-2}

\usepackage{xcolor}
\definecolor{WeakGreen}{HTML}{52BD7C}
\definecolor{Red}{HTML}{D7333B}
\definecolor{Green}{HTML}{00973E}
\definecolor{XanaBlue}{HTML}{4D53C8}

\usepackage{listings}
\usepackage{pythonhighlight}
\usepackage[colorlinks=true,linkcolor=Red,citecolor=Green,urlcolor=Red]{hyperref}
\usepackage{footnotebackref}
\usepackage{booktabs}
\usepackage{graphicx}
\newcommand{\cnot}{\operatorname{CNOT}}
\newcommand{\cz}{\operatorname{CZ}}
\newcommand{\cy}{\operatorname{CY}}
\newcommand{\ciy}{\operatorname{C}\!i\!\operatorname{Y}}

\usepackage{here}
\usepackage{bbold}
\usepackage{quantikz}

\newcommand{\block}[1][2]{\gategroup[#1,steps=2,style={inner sep=0pt,rounded corners, color=Red}]{}}
\newcommand{\emptyblock}{\gate[2]{\ \ \ }\gategroup[2,steps=1,style={inner sep=0pt,rounded corners, color=Red,fill=white}]{}}

\usepackage{amsthm}
\usepackage{thmtools,mathtools}
\usepackage{thm-restate}
\usepackage[capitalize]{cleveref}
\newtheorem{thm}{Theorem}
\crefname{thm}{Thm.}{Thms.}

\crefname{lemma}{Lemma}{Lemmas}

\newtheorem{definition}{Definition}
\crefname{definition}{Def.}{Defs.}

\crefname{remark}{Remark}{Remarks}

\crefname{prop}{Prop.}{Props.}

\newtheorem{conjecture}[thm]{Conjecture}
\crefname{conjecture}{Conj.}{Conjs.}

\crefname{section}{Sec.}{Secs.}
\crefname{appendix}{App.}{Apps.}
\crefname{table}{Tab.}{Tabs.}

\newcommand{\C}{\mathbb{C}}
\newcommand{\R}{\mathbb{R}}
\newcommand{\Z}{\mathbb{Z}}
\newcommand*{\id}{\mathchoice
  {\openone}
  {\openone}
  {\scalebox{.7}{\openone}} 
  {\scalebox{.5}{\openone}} 
}
\newcommand{\groupG}{\mathcal{G}}
\renewcommand{\det}[1]{\mathrm{det}(#1)}

\begin{document}
\title{Unitary synthesis with optimal brick wall circuits}
\author{David Wierichs}
\affiliation{Xanadu, Toronto, ON, M5G 2C8, Canada}

\author{Korbinian Kottmann}
\affiliation{Xanadu, Toronto, ON, M5G 2C8, Canada}

\author{Nathan Killoran}
\affiliation{Xanadu, Toronto, ON, M5G 2C8, Canada}

\begin{abstract}
    We present quantum circuits with a brick wall structure using the optimal number of parameters and two-qubit gates to parametrize $SU(2^n)$, and provide evidence that these circuits are universal for $n\leq 5$.
    For this, we successfully compile random matrices to the presented circuits and show that their Jacobian has full rank almost everywhere in the domain.
    Our method provides a new state of the art for synthesizing typical unitary matrices from $SU(2^n)$ for $n=3, 4, 5$, and we extend it to the subgroups $SO(2^n)$ and $Sp^\ast(2^n)$.
    We complement this numerical method by a partial proof, which hinges on an open conjecture that relates universality of an ansatz to it having full Jacobian rank almost everywhere.
\end{abstract}

\maketitle

\section{Introduction}

Decompositions of multi-qubit quantum gates into more elementary operations are a persistent theme in quantum computing research across decades~\cite{Khaneja-Glaser,Shende-Bullock-Markov,Shende-Markov-Bullock,Vatan-Williams,Shende-Markov,Krol-Al-Ars,Rakyta-Zimboras,wierichs2025recursive}.
For the particular example of decomposing an arbitrary three-qubit operation into single-qubit rotation gates and $\cnot$ (equivalently, $\cz$ gates), the best analytic decomposition consists of $19$ $\cnot$ gates and a minimum of $63$ rotation gates~\cite{Krol-Al-Ars}.
Also see~\cite{Krol-Al-Ars} for a nice overview of what these analytical methods have achieved.

Optimization-based methods have been reported to achieve general decompositions with $15$ $\cnot$ gates~\cite{Rakyta-Zimboras}, and modern synthesis tools are able to create even better target-specific decompositions~\cite{Younis}, extending the few analytic decompositions known for special three-qubit gates such as the Toffoli ($6$ $\cnot$ gates).
The theoretical lower bound of the $\cnot$ count $c$ for $n$ qubits is given by $c\geq\left\lceil\frac{1}{4}(4^n-3n-1)\right\rceil$~\cite{shende2003minimal}, yielding the lower bound $c\geq 14$ for $n=3$.

Here, we present circuit templates with a brick wall structure that are optimal in both the two-qubit gate count and parameter count, which we conjecture to be universal. These brick wall circuits form our main result, described in \cref{sec:main_results}. 
We consider it instructive to present the thoughts and motivation that went into finding these circuit templates in the first place, and describe them in \cref{sec:derivation}.
We have not yet found an analytic formula or linear algebraic routine computing the circuit parameters from the matrix elements, but we show in \cref{sec:numerical_results} how to employ these templates in practice by compiling them to different classes of unitary matrices via variational optimization.
Together with further numerical evidence, these results lead us to conjecture universality of the presented circuit families.
For \textit{typical} representatives of $SU(2^n)$ on $n=3, 4, 5$ qubits, the presented brick wall circuits thus improve on the best known circuits, both in terms of two-qubit and parameter count.
\Cref{sec:numerical_results} also provides context for applications in which the presented circuits may prove useful.

The number of rotation gates is particularly important for fault-tolerant quantum computing (FTQC) because rotation gates with non-Clifford angles are the main bottleneck in most quantum error correction schemes~\cite{Litinski2019}.
We therefore emphasize the number of non-Clifford rotation gates in our numerical studies and view it as the main figure of merit. Motivated by earlier generations of quantum hardware, optimizing the number of two-qubit gates has in the past been a main focus in the literature on unitary synthesis, and it can be considered a kind of bonus when compiling for FTQC.

We provide further mathematical considerations regarding the universality of these circuits in \cref{sec:math}, leaving a complete proof open.
Beyond generic $SU(2^n)$ matrices, we also provide templates for the subgroups $SO(2^n)$ and $Sp^\ast(2^n)$\footnote{By $Sp^\ast(2^n)$ we denote the symplectic subgroup of the special unitary group $SU(2^n)$. This is often denoted $Sp(2^{n-1})$ in the literature. Because this can be confusing compared to the notation for $SU$ and $SO$, we adopt a simpler, unified form.} in \cref{app:SON} and \cref{app:SpN}, respectively.

\section{Main results}
\label{sec:main_results}

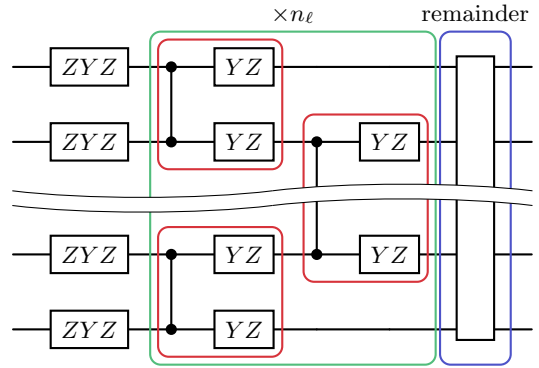
\begin{figure}[ht]
\begin{quantikz}
    & \gate{ZYZ} & \ctrl{1}\gategroup[2,steps=2,style={inner sep=0pt,rounded corners, color=Red}]{}\gategroup[5,steps=4,style={inner sep=3pt,rounded corners,color=WeakGreen}]{$\times n_\ell$} & \gate{Y Z} & & & \gate[5]{\ \ }\gategroup[5,steps=1,style={inner sep=3pt,rounded corners,color=XanaBlue}]{remainder} & \\
    & \gate{ZYZ} & \ctrl{0} & \gate{Y Z} & \ctrl{1}\gategroup[3,steps=2,style={inner sep=0pt,rounded corners, color=Red}]{} & \gate{Y Z} & & \\
    \wave & & & & & & & \\
    & \gate{ZYZ} & \ctrl{1}\gategroup[2,steps=2,style={inner sep=0pt,rounded corners, color=Red}]{} & \gate{Y Z} & \ctrl{-1} & \gate{Y Z} & & \\
    & \gate{ZYZ} & \ctrl{0} & \gate{Y Z} & & & &
\end{quantikz}
\caption{Brick wall circuit structure with the optimal number of parameters and two-qubit gates, which we claim to yield a universal circuit. The number of times, $n_\ell$, that the green box is repeated is given in \cref{eq:nell}. The ``remainder" block is constructed such that the minimal number of parameters is used. \cref{tab:SUN} shows the explicit structure for $n\in \{3, 4, 5\}$.}
\label{fig:sun_circuit}
\end{figure}

\begin{table*}[t]
\centering
\begin{tabular}{l|c|c|c}
$n$ & $3$ & $4$ & $5$ \\
\hline
circuit & 
\scalebox{0.6}{
\begin{quantikz}
    & \gate{ZYZ} & \ctrl{1}\block\gategroup[3,steps=4,style={inner sep=3pt,rounded corners,color=WeakGreen}]{$\times 6$} & \gate{Y Z} & & & \ctrl{1}\block\gategroup[3,steps=4,style={inner sep=3pt,rounded corners,color=XanaBlue}]{remainder} & \gate{Y Z} & & & \\
    & \gate{ZYZ} & \ctrl{0} & \gate{Y Z} & \ctrl{1}\block & \gate{Y Z} & \ctrl{0} & \gate{Y Z} & \ctrl{1} & \gate{Y} & \\
    & \gate{ZYZ} & & & \ctrl{0} & \gate{Y Z} & & & \ctrl{0} & \gate{Y} &
\end{quantikz}
}
&
\scalebox{0.6}{
\begin{quantikz}
    & \gate{ZYZ} & \ctrl{1}\block\gategroup[4,steps=4,style={inner sep=3pt,rounded corners,color=WeakGreen}]{$\times 20$} & \gate{Y Z} & & & \ctrl{1}\gategroup[4,steps=2,style={inner sep=3pt,rounded corners,color=XanaBlue}]{remainder}  & \gate{Y Z} & \\
    & \gate{ZYZ} & \ctrl{0} & \gate{Y Z} & \ctrl{1}\block & \gate{Y Z} & \ctrl{0} & \gate{Y} & \\
    & \gate{ZYZ} & \ctrl{1}\block & \gate{Y Z} & \ctrl{0} & \gate{Y Z} & & & \\
    & \gate{ZYZ} & \ctrl{0} & \gate{Y Z} & & & & &
\end{quantikz}
}
&
\scalebox{0.6}{
\begin{quantikz}
    & \gate{ZYZ} & \ctrl{1}\block\gategroup[5,steps=4,style={inner sep=3pt,rounded corners,color=WeakGreen}]{$\times 63$} & \gate{Y Z} & & & \\
    & \gate{ZYZ} & \ctrl{0} & \gate{Y Z} & \ctrl{1}\block & \gate{Y Z} & \\
    & \gate{ZYZ} & \ctrl{1}\block & \gate{Y Z} & \ctrl{0} & \gate{Y Z} & \\
    & \gate{ZYZ} & \ctrl{0} & \gate{Y Z} & \ctrl{1}\block & \gate{Y Z} & \\
    & \gate{ZYZ} & & & \ctrl{0} & \gate{Y Z} &
\end{quantikz}
}

\\
\# parameters & $3 \cdot 3 + 6 \cdot 4 \cdot 2 + 4 + 2 = 63$  & $4\cdot3 + 20 \cdot 4 \cdot 3 + 3 = 255$ & $5\cdot3 + 63 \cdot 4 \cdot 4 = 1023$ \\
\# CZs & $6 \cdot 2 + 2 = 14$ & $20\cdot3 + 1 = 61$ & $63\cdot4 = 252$ \\
\end{tabular}
\caption{Brick wall circuits for $SU(2^n)$ on $n=3, 4, 5$ qubits with optimal parameter count ($4^n-1$) and number of $\cz$ gates ($\left\lceil\frac{1}{4}(4^n-3n-1)\right\rceil$). We conjecture these circuits to be universal and provide numerical evidence supporting this claim. A constructive rule for these circuits is to stack bricks (red) into $n_\ell$ complete layers (green) and a remainder layer (blue), which may contain an incomplete brick; also see \cref{fig:sun_circuit}.}
\label{tab:SUN}
\end{table*}
Our main contribution consists of three families of circuits, for $SU(2^n)$, $SO(2^n)$, and $Sp^\ast(2^n)$, respectively, that are optimal in both parameter and two-qubit-gate count, and which we conjecture to be universal. 
We begin with $SU(2^n)$ and briefly present the other two cases, leaving details to \cref{app:SON,app:SpN}, respectively.
For $SU(2^n)$, a universal circuit with optimal parameter and $\cz$ gate count is given in \cref{fig:sun_circuit} (note that this circuit is one of many possibilities, it just happens to have the particularly regular brick wall structure).
The ansatz starts with a layer of $SU(2)$ gates on each qubit, here parametrized by an Euler decomposition in terms of $R_Z(\alpha) R_Y(\beta) R_Z(\gamma)$. We abbreviate this as a $Z Y Z$ box. The main body of the ansatz consists of red blocks, or bricks, each composed of a single $\cz$ gate followed again by $SU(2)$ gates on each qubit. Because we can pull a $Z$ rotation through the $\cz$ and merge it with the previous rotation, each $SU(2)$ gate in the red block is reduced to just $R_Z(\alpha_j) R_Y(\beta_j)$\footnote{Keep in mind that the order of circuit diagrams and matrix multiplication are reversed.}. The red blocks are arranged in a brick wall layout and are overall repeated 
\begin{align}\label{eq:nell}
    n_\ell = \left\lfloor \frac{4^n - 1 - 3n}{4(n-1)}\right\rfloor
\end{align}
times. Here, $d=4^n - 1$ is the total number of parameters, $3n$ is the parameter count in the initial layer and $4(n-1)$ the number of parameters per red block times the number of red blocks per green block.
As we can see, this circuit is compatible with linear nearest-neighbour connectivity. 

The ``remainder" block is such that the number of parameters reaches the desired dimension of the group, i.e., $d=4^n-1$ for $SU(2^n)$. We show the explicit circuit structures for $n=3, 4, 5$ in \cref{tab:SUN}. A constructive rule that we find to work is to keep packing parameters per two-qubit gate as tightly as possible and simply stop when the number of parameters reaches the dimension of the group.

We use $\cz$ as the entangling gate as it is symmetric in terms of the qubits it is acting on, which was helpful for the discovery of these circuits detailed in \cref{sec:derivation_B}. If one prefers to work with $\cnot$ gates, we can transform $\cz_{ij} = H_j \cnot_{ij} H_j$ and then merge the resulting Hadamard gates with the surrounding single-qubit gates. This yields the same circuit structure as presented in \cref{fig:sun_circuit}, but simply exchanging $\cz$ with $\cnot$ gates.

Recall that the lower bound for the required number of two-qubit gates in an $n$-qubit circuit is given by~\cite{shende2003minimal}

\begin{equation}
\label{eq:lower_bound_2q_su}
    c \geq \left\lceil \frac{1}{4} \left(4^n - 1 - 3n \right)\right\rceil.
\end{equation}
This bound can be derived from simple dimension-counting considerations. As depicted in \cref{fig:sun_circuit}, we can only ever add $4$ parameters per two-qubit gate. Further, we can always start by performing arbitrary $SU(2)$ gates on each qubit, yielding $3$ free parameters per qubit. So for $c$ two-qubit gates, we have $3n + 4c$ parameters overall. This number needs to at least match the dimension of the group, which in the case of $SU(2^n)$ is $4^n -1$. Thus, we require $3n + 4c \geq 4^n - 1$ and arrive at the lower bound in \cref{eq:lower_bound_2q_su}. 
The circuits presented in \cref{tab:SUN} are optimal in both parameter and two-qubit gate counts. We conjecture them to be universal and provide strong numerical evidence for this in \cref{sec:numerical_results}.

We can generalize the bound in \cref{eq:lower_bound_2q_su} to other groups with dimension $d$ for $n$ qubits:
\begin{equation}\label{eq:lower_bound_2q_general}
    c \geq\left\lceil \frac{1}{n_\text{params/2q}} \left(d - n_\text{initial} \right)\right\rceil,
\end{equation}
where $n_\text{params/2q}$ is the number of parameters per two-qubit gate that we can add ($4$ in the case above) and $n_\text{initial}$ is the number of parameters in the initial layer that has no two-qubit gates ($3n$ in the case above).
These circuits can also be readily translated to so-called Pauli product rotations (PPRs) \cite{PPR,Litinski2019}, as detailed in ~\cref{app:pprs}.

For the orthogonal group $SO(2^n)$, we have $n_\text{params/2q}=2$, $d=\tfrac{1}{2}2^n(2^n-1)$, and $n_\text{initial}=n$, leading to the lower bound
\begin{equation}\label{eq:lower_bound_2q_so}
    c \geq\left\lceil \frac{1}{2} \left(\frac{1}{2}2^n(2^n-1) - n \right)\right\rceil;
\end{equation}
see \cref{app:SON} for details. The minimal circuit structure, which we conjecture to be universal, is a simple restriction of the $SU(2^n)$ brick wall to the reals:

\begin{center}
\begin{quantikz}
    & \gate{Y} & \ctrl{1}\gategroup[2,steps=2,style={inner sep=0pt,rounded corners, color=Red}]{}\gategroup[5,steps=4,style={inner sep=3pt,rounded corners,color=WeakGreen}]{$\times n_\ell$} & \gate{Y} & & & \gate[5]{\ \ }\gategroup[5,steps=1,style={inner sep=3pt,rounded corners,color=XanaBlue}]{remainder} & \\
    & \gate{Y} & \ctrl{0} & \gate{Y} & \ctrl{1}\gategroup[3,steps=2,style={inner sep=0pt,rounded corners, color=Red}]{} & \gate{Y} & & \\
    \wave & & & & & & & \\
    & \gate{Y} & \ctrl{1}\gategroup[2,steps=2,style={inner sep=0pt,rounded corners, color=Red}]{} & \gate{Y} & \ctrl{-1} & \gate{Y} & & \\
    & \gate{Y} & \ctrl{0} & \gate{Y} & & & &
\end{quantikz}
\end{center}

For the symplectic group $Sp^\ast(2^n)$, we have at most $n_\text{params/2q}=3$, $d=\tfrac{1}{2}2^n(2^n+1)$ and $n_\text{initial}=n+2$, leading to the lower bound
\begin{equation}\label{eq:lower_bound_2q_sp}
    c \geq\left\lceil \frac{1}{3} \left(\frac{1}{2}2^n(2^n+1) - (n+2) \right)\right\rceil;
\end{equation}
see \cref{app:SpN} for details.
The template for this group exclusively uses entanglers with the first qubit to achieve $n_\text{params/2q}=3$:

\hspace{-0.6cm}
\scalebox{0.73}{
\begin{quantikz}
& \gate{ZYZ} & \gate{iY} \block\gategroup[5,steps=7,style={inner sep=3pt,rounded corners,color=WeakGreen}]{$\times n_\ell$} & \gate{Z Y} & \gate{iY}\block[3] & \gate{Z Y} &  \ \ldots\ & \gate{iY}\block[5] & \gate{ZY} & \gate[5]{\ \ }\gategroup[5,steps=1,style={inner sep=3pt,rounded corners,color=XanaBlue}]{remainder} & \\
& \gate{Y} & \ctrl{-1} & \gate{Y} & & & \ \ldots\ & & & & \\
& \gate{Y} & & & \ctrl{-2} & \gate{Y} & \ \ldots\ & & & & \\
\wave & & & & & & & & & & \\
& \gate{Y}& & & & & \ \ldots\ & \ctrl{-4} & \gate{Y}& &
\end{quantikz}%
}

\section{Derivation}
\label{sec:derivation}
Here we will outline how we arrived at the universal three-qubit circuit template via an exhaustive search, which was feasible only due to the small qubit count. However, based on the regular structure of the circuit, we then generalized the circuit to higher qubit counts and other groups, relying on numerical verification tools for universality.
We believe that sharing our thought process may be instructive for similar future work and hope that the reader finds it valuable.
The sobering key realization for our approach is that the space of circuits with 14 entangling gates (after filtering and reduction via symmetries) is small enough to simply search it exhaustively and identify candidates for universality by testing necessary conditions on each ansatz.

The established choice in unitary synthesis techniques for the static entangler is the $\cnot$ gate.
However, the symmetry of the $\cz$ gate makes it a nicer gate to work with, already reducing the search space size.
We will proceed in the following steps:

\begin{enumerate}
    \item[\ref{sec:ansatz_characterization}] Characterize the search space of all $SU(2^n)$ circuits with a given $\cz$ count.
    \item[\ref{sec:derivation_B}] Reduce the size of the search space with equivalence transformations, or by discarding less favourable solutions a priori.
    \item[\ref{sec:derivation_rank_test}] Numerically check a necessary condition for each ansatz in the reduced search space.
    \item[\ref{sec:selecting_candidates}] Select preferred candidate(s), minimize their parameter count, and test additional necessary conditions.
\end{enumerate}

The full workflow we used to discover the presented universal circuits is available in~\cite{our_repo}.

\subsection{Ansatz characterization}\label{sec:ansatz_characterization}

A fully general three-qubit ansatz with a given $\cz$ count $c$ is fully specified by $c$ pairs of qubits on which the $\cz$ gates act, without ordering \textit{within} the pairs.
Such an ansatz begins with an arbitrary single-qubit rotation on all three qubits, and then executes one two-qubit block, or brick, at a time, consisting of a $\cz$ gate followed by an $R_Y$ and an $R_Z$ rotation on both qubits.
This ansatz structure using two-qubit bricks is also used, e.g., in~\cite{Rakyta-Zimboras}, and forms the basis of the lower bound for $c$ from~\cite{Shende-Markov-Bullock,shende2003minimal}; also see \cref{sec:main_results}.
For $c$ entangling gates, the circuit contains $9+4c$ rotation angles---or parameters.

\subsection{Reducing the search space}\label{sec:derivation_B}

Our method to find a universal ansatz of the above form is rather simple: we essentially try all possibilities. Each $\cz$ on $n$ qubits is characterized by one of $\tfrac12n(n-1)$ possible unordered pairs of qubit indices, giving rise to $(\tfrac12(n^2-n))^c$ distinct circuits with $c$ $\cz$ gates.
For our original target of finding an $n=3$ qubit circuit with $c=14$ $\cz$s, this results in $3^{14}\approx 4.8\cdot 10^{6}$ different circuits.
While this search space certainly becomes unreasonably large for qubit counts above three, we can reduce it to an easily-managed size for $n=3, c=14$, using the following steps.

\begin{enumerate}
    \item Only consider nearest-neighbour entangling gates,~i.e., those acting on the qubit pairs $\{0, 1\}$ and $\{1, 2\}$ but not on $\{0, 2\}$. This is not an equivalence transformation of the search space, but we simply employ some optimism and hope that a valid decomposition will exist with these connectivity constraints. This reduces the number of distinct entanglers from $\tfrac12 n(n-1)$ to $n-1$,~i.e., from $3$ to $2$.
    \item Each specification that contains a sequence of four identical qubit pairs can be neglected. This is because such a sequence can be merged into a general two-qubit operation with three $\cz$ gates, lowering the total $\cz$ count from $c=14$ to $c=13$, which is below the theoretical lower bound and thus can't describe a universal ansatz.
\end{enumerate}

At each step, the number of ansatz specifications to test is

\begin{align}
    3^{14}
    \overset{\text{1.}}{\longrightarrow} 2^{14}
    \overset{\text{2.}}{\longrightarrow} 6272.
\end{align}
We thus reduced the search space by a factor of about $760$.
Even though it might be considered an advantage for an ansatz to only require the weaker connectivity imposed in step 1, it is feasible to skip this step and still search the resulting space (containing $3.5\cdot 10^6$ ans\"{a}tze) exhaustively, in a few hours instead of a minute.

There are additional reductions one could perform, which did not seem necessary for three qubits but might be useful for larger qubit counts or other generalizations.
For example, reversing the sequence of qubit pairs maintains universality, as do permutations of the qubit indices\footnote{If the connectivity is restricted, as is done in step 1 above, only some permutations maintain this connectivity.}.

\subsection{Necessary condition: Jacobian rank}
\label{sec:derivation_rank_test}
For each ansatz in the search space, we run a test based on a necessary condition for the rank of its Jacobian matrix.
We randomly sample $n_\text{initial}+ n_\text{params/2q}\ c=9+4c=65$ parameters $\theta$ and compute the Jacobian $J_F(\theta)\in\C^{65\times 8\times 8}$ of the matrix function $F: \R^{65}\to \C^{8\times 8}$ of the ansatz at $\theta$.
In order for the ansatz to be universal, the resulting $65$ matrices with shape $8\times 8$ must span the special unitary algebra, $\mathfrak{su}(8)$, almost everywhere across the parameter space.
This can be checked by verifying that the rank of the Jacobian is $\operatorname{dim}(\mathfrak{su}(8))=63$.
If this is the case, we know that locally around $F(\theta)$ the ansatz can produce any other unitary, as it spans the full tangent space of $SU(8)$ at $F(\theta)$. A similar technique is explored to produce piecewise parametrizations in~\cite{Funcke}.
Note that the number of parameters is larger than the maximal rank, indicating that the circuits are not parameter-optimal yet. We will fix this in \cref{sec:opt_angle_count}.

As we will show in \cref{sec:math}, $J_F$ will have full rank almost everywhere on its domain if it has full rank at \emph{any} $\theta$.
Even though the analytical probability of sampling a critical point $\theta^\ast$ of $F$ vanishes in this case, there is a nonzero probability due to finite numerical precision to sample parameters at which $J_F(\theta)$ is rank-deficient.
Our test thus will discard a perfectly valid ansatz if this happens.
For our purposes, this is not a problem as long as we find (sufficiently many) circuits that pass the test. If for some reason \textit{all} universal circuits are sought, the chance for such a false negative can be suppressed by repeating the test with new random parameters.

The test above requires the computation of the Jacobian $J_F(\theta)$, which can be performed conveniently using backpropagation. In our computational pipeline we used tools from PennyLane~\cite{pennylane} with JAX~\cite{jax} and its automatic differentiation and just-in-time (JIT) compilation features. This allows us to run the Jacobian rank test for all $6272$ circuits in just over a minute on a MacBook Air (M4).
We found $1084$ circuits with linear connectivity passing the Jacobian rank test~\cite{our_repo}.

\subsection{Selecting universal candidates}\label{sec:selecting_candidates}
None of the $1084$ solutions that we found to pass the Jacobian rank test use three subsequent bricks on the same qubit pair.
Out of the solutions, we select two with a particularly regular structure.
For the first candidate, we aim to group gates into largest-possible two-qubit blocks made up of multiple bricks, resulting in the following circuit:

\begin{center}
\resizebox{\linewidth}{!}{
\begin{quantikz}
    & \gate{ZYZ} & \ctrl{1}\block\gategroup[3,steps=5,style={inner sep=3pt,rounded corners,color=WeakGreen}]{$\times 3$} & \gate{YZ} & \emptyblock & & & \emptyblock & \emptyblock & \\
    & \gate{ZYZ} & \ctrl{0} & \gate{YZ} & & \emptyblock & \emptyblock & & & \\
    & \gate{ZYZ} & & & & & & & &
\end{quantikz}\!,
}
\end{center}
where we marked arbitrary $SU(2)$ operators in the initial layer with the Euler decomposition $ZYZ$ and the subsequent single-qubit gates consist of two single-parameter rotations as discussed before.
We show this candidate for expository purposes only, and for the second candidate we aim for an even more reduced structure. We will predominantly use this second candidate throughout this manuscript, generalizing it to higher qubit counts and the other classic Lie groups $SO(2^n)$ and $Sp^\ast(2^n)$:

\begin{center}
\scalebox{0.9}{
\begin{quantikz}
    & \gate{ZYZ} & \ctrl{1}\gategroup[3,steps=3,style={inner sep=3pt,rounded corners,color=WeakGreen}]{$\times 7$}\block & \gate{YZ} & & \\
    & \gate{ZYZ} & \ctrl{0} & \gate{YZ} & \emptyblock & \\
    & \gate{ZYZ} & & & & 
\end{quantikz}
}\!,
\end{center}
with the same types of single-qubit gates as before.

All other circuits presented in \cref{sec:main_results} are derived from this prototypical brick wall circuit, and verified with the rank test and further numerical evidence presented in \cref{sec:numerical_results}.

\subsection{Minimizing the parameter count}\label{sec:opt_angle_count}

Our characterization of $SU(2^n)$ circuits in \cref{sec:ansatz_characterization} yields three-qubit circuits with $9+4c=65$ parameters. However, we know that $63$ degrees of freedom are sufficient to parametrize $SU(8)$, so that we may aim to reduce the parameter count by two.
It turns out that for the two circuits we selected above, deleting any two gates at the end maintains the Jacobian rank, as long as it does not lead to further gate cancellations\footnote{And thus to parameter reductions that would push the parameter count below the theoretical minimum of $63$.}.
Examples for the resulting parameter-optimal circuits are given by 
\begin{center}
\scalebox{0.78}{
\begin{quantikz}
    & \gate{ZYZ} & \ctrl{1}\block\gategroup[3,steps=5,style={inner sep=3pt,rounded corners,color=WeakGreen}]{$\times 3$} & \gate{YZ} & \emptyblock & & &
    \emptyblock & \ctrl{1} & \gate{Y} & \\
    & \gate{ZYZ} & \ctrl{0} & \gate{YZ} & & \emptyblock & \emptyblock & & \ctrl{0} & \gate{Y} & \\
    & \gate{ZYZ} & & & & & & &&&
\end{quantikz}
}\!,
\end{center}
and by 
\begin{center}
\scalebox{0.78}{
\begin{quantikz}
    & \gate{ZYZ} & \ctrl{1}\gategroup[3,steps=3,style={inner sep=3pt,rounded corners,color=WeakGreen}]{$\times 6$}\block & \gate{YZ} & & \emptyblock\gategroup[3,steps=3,style={inner sep=3pt,rounded corners,color=XanaBlue}]{remainder}  & & & \\
    & \gate{ZYZ} & \ctrl{0} & \gate{YZ} & \emptyblock & & \ctrl{1} & \gate{Y} &\\
    & \gate{ZYZ} & & & & & \ctrl{0} & \gate{Y} &
\end{quantikz}
}.
\end{center}

\section{Numerical results}
\label{sec:numerical_results}

We have so far constructed circuits with optimal resources that we conjecture to be universal based on the Jacobian rank test described in \cref{sec:derivation_rank_test}. 
In the next section (\cref{sec:expressibility}), further evidence for universality is provided in the form of an expressibility measure that we compute numerically.

We then move on to using our circuits for quantum compilation in practice. To this end, unitary synthesis for random unitary matrices is performed as a proof of principle in \cref{sec:unitary_synthesis_typical}, and applied to vibronic and fast-forwardable Hamiltonians as relevant real-world applications in \cref{sec:applications}.

\subsection{Expressibility by Sim et al.}
\label{sec:expressibility}

\begin{figure}
    \centering
    \includegraphics[width=\linewidth]{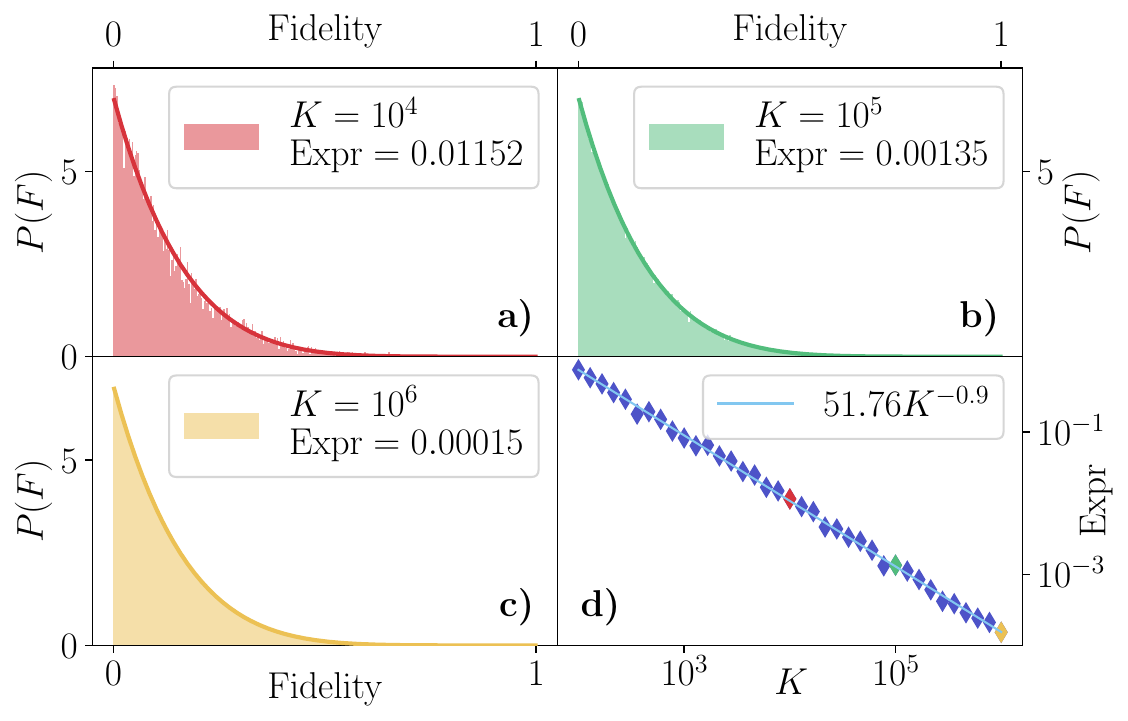}
    \caption{Numerical expressibility test from~\cite{Sim-Johnson-Aspuru-Guzik} for the parameter-optimal circuit on three qubits. 
    We randomly sample $K$ pairs of parameter vectors, compute the fidelities between the states prepared by the ansatz, and compare the resulting fidelity distribution to that of Haar random unitaries. The expressibility $\operatorname{Expr}$ of the ansatz is defined as the Kullback-Leibler divergence between them, with smaller values indicating larger expressibility.
    \textbf{a-c)} Fidelity distributions for the circuit (histogram) and for Haar random unitaries (solid line), for $K=10^4,10^5,10^6$, respectively.
    \textbf{d)} Expressibility $\operatorname{Expr}$ for $10^2\leq K\leq 10^6$ (markers). The numerical value depends on $K$, and we observe a steady convergence towards $0$ with scaling $\mathcal{O}(K^{-0.9})$ (solid line), which we interpret as an indication for universality. See \cref{fig:expressibility_app} for the cases $n=4, 5$.}
    \label{fig:expressibility}
\end{figure}

In~\cite{Sim-Johnson-Aspuru-Guzik}, an ansatz expressibility measure based on the Kullback-Leibler divergence was proposed. This measure is given by
\begin{align}
    \operatorname{Expr} = D_{\text{KL}}\left(P_\text{ansatz}(f)\| P_\text{Haar}(f)\right),
\end{align}
with $D_\text{KL}$ denoting the Kullback-Leibler divergence and $P_\text{ansatz}$ and $P_\text{Haar}$ denoting the probability distribution of the fidelities ($f$) for the ansatz and the Haar random distribution, respectively.
We can estimate $\operatorname{Expr}$ numerically for a finite sample count $K$ by sampling $K$ pairs of parameter vectors and estimating the distribution of fidelities $P_\text{ansatz}(F)$ from the empirical distribution that arises from the pairs of states prepared with these parameters.
The Haar fidelity distribution is known analytically to be $P_\text{Haar}(f) = (2^n-1)(1-f)^{2^n-2}$, which can be discretized to compute the divergence with respect to $P_\text{ansatz}(f)$ and thus $\operatorname{Expr}$.

If the measure satisfies $\operatorname{Expr}\to 0$ for $K\to\infty$, this indicates that the candidate circuit in question is as expressible as the Haar random distribution. Sim et al. leave it open to prove that a vanishing $\operatorname{Expr}$ implies \textit{equality} to the Haar distribution.
We run the numerical analysis for the second optimized circuit from \cref{sec:opt_angle_count} with $10^2\leq K\leq 10^6$ samples and a discretization into $300$ bins for the interval $[0, 1]$.
In \cref{fig:expressibility}, we show the obtained distributions for $K\in\{10^4, 10^5, 10^6\}$ and how $\operatorname{Expr}$ vanishes with increasing $K$.
We interpret this as further evidence for the universal expressibility of the ansatz.
The code for this test is available in~\cite{our_repo}.

\subsection{Unitary synthesis of typical matrices}
\label{sec:unitary_synthesis_typical}

In this subsection, we demonstrate the practical usage of our circuits as universal parametrizations of unitary matrices. To this end, we sample Haar random unitary matrices $U$ from $SU(2^n)$ and compile them to the circuit ansatz. This process of going from a unitary matrix to a circuit is typically referred to as unitary synthesis. Unfortunately, we do not have analytic formulas or linear-algebra-based constructions to obtain the parameters of the ansatz directly from the matrix $U$. We therefore resort to variationally fitting the parameters via gradient descent of the cost function

\begin{equation}\label{eq:cost_func}
    \mathcal{L}(\theta) := 1 - \frac{1}{2^n}\left|\text{tr}\left(U_\text{target}^\dagger V(\theta)\right) \right|.
\end{equation}

We start the optimization from random initializations, sampled from a normal distribution around $0$. Because this can sometimes get stuck in local minima, we restart the process if we do not reach the desired accuracy $\epsilon$ for $\mathcal{L}(\theta)$ within a user-set number of epochs, typically chosen heuristically for the size of the matrix. We then repeat this process until convergence is reached, compiling all target matrices successfully for the accuracy $\epsilon = 10^{-10}$ that we used in all applications reported here.

This unitary synthesis method is implemented as a plug-and-play function \texttt{unicirc.compile} in our Python package \pyth{unicirc}~\cite{our_repo}. The repository also contains code with which all figures of this paper can be readily reproduced. We use PennyLane~\cite{pennylane}, JAX~\cite{jax} and Optax~\cite{optax} for our implementations.

In \cref{fig:convergence} we show the cost function \cref{eq:cost_func} during the optimization of $10$ random target unitaries sampled from the unitary group, starting from $10$ random initializations (drawn from a normal distribution around $0$, i.e., \pyth{jax.random.normal} with default values). We use the \pyth{optax.lbfgs} optimizer with \pyth{memory_size=5*d} and \pyth{learning_rate=None}, where $d=4^n-1$ is the parameter count. For each target, at least one attempt converges, and we show all $10\times 10$ runs in the same plot for brevity (see~\cite{our_repo} for the individual plots for each target unitary). We observe that roughly half of the attempts succeed for the chosen convergence threshold of $\epsilon = 10^{-10}$. In particular, for smaller qubit numbers the optimization often gets stuck in local minima. It seems that these local minima issues improve with increasing qubit numbers, in terms of the accuracy they produce. This leads to the somewhat counterintuitive picture painted in \cref{fig:successprob}, showing that the success rate $\kappa$ for larger threshold values actually increases as we increase the number of qubits.

\begin{figure}
    \centering
    \includegraphics[width=\linewidth]{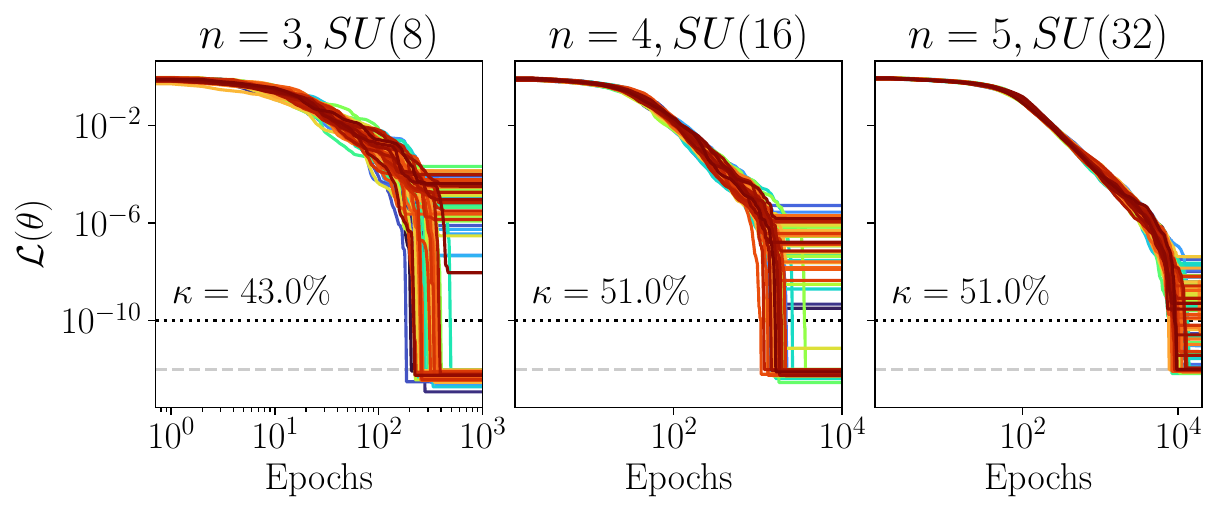}
    \caption{Compiling $10$ Haar random unitaries by starting variational optimization from $10$ random initial values each, sampled from a normal distribution with $\mu=0, \sigma=\tfrac15$. Each target unitary converges at least once within the $10$ trials to the threshold of $10^{-10}$ (dotted line). $\kappa$ in the lower left corner indicates the overall success rate over the $10\times10$ trials. The used target precision is $10^{-12}$ (dashed line), allowing for the more detailed analysis of $\kappa$ in \cref{fig:successprob}.}
    \label{fig:convergence}
\end{figure}

\begin{figure}
    \centering
    \includegraphics[width=\linewidth]{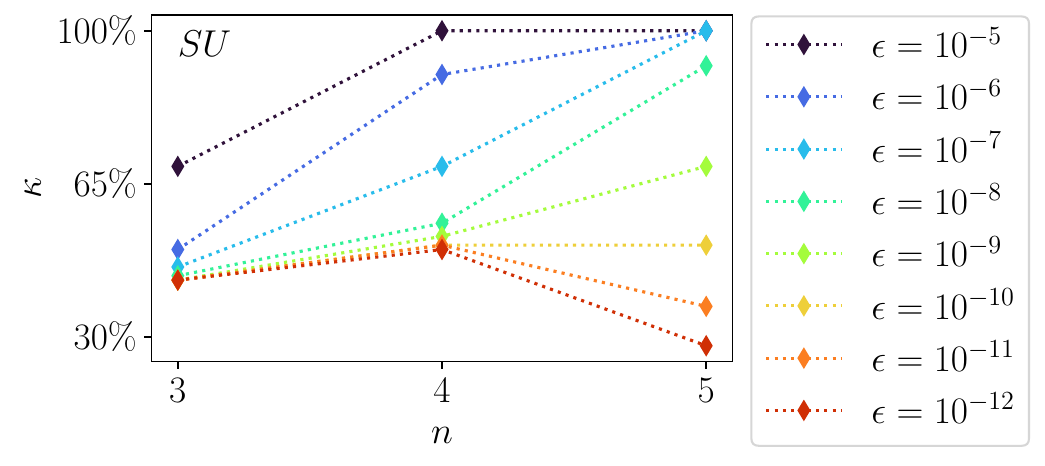}
    \caption{Success rates $\kappa$ from the $100$ compilation attempts in \cref{fig:convergence} for different threshold values $\epsilon$. Curiously, for threshold values above $10^{-10}$, the success probability increases with the qubit number. One potential explanation is that the local minima encountered for higher qubit counts provide approximations with higher accuracy.}
    \label{fig:successprob}
\end{figure}

From a practical perspective, the different compilation attempts starting from distinct initial parameters are not crucial. Rather, we are interested in a single practically relevant metric that combines the success rate, the required optimization epochs, and the computational cost per epoch; for this, we choose the time required to synthesize a given unitary, including all failed attempts that converged to local minima and had to be restarted.
To evaluate this metric, we sampled $1000$ ($500$, $100$) random unitaries for $n=3$ ($n=4$, $n=5$) and ran repeated optimizations for each of them until the first convergence occurred.
The resulting compilation times are shown in \cref{fig:compilation_times}, the cost during optimization is shown in \cref{fig:compilation_costs}.
For the special orthogonal group $SO(2^n)$ and the symplectic unitary group $Sp^\ast(2^n)$, we show the analogous experiments in \cref{fig:convergence_success_rates_so,fig:convergence_success_rates_sp}. The compilation is similarly successful, with the exception of some target matrices in $SO(8)$, that were not compiled successfully.

\begin{figure}  
    \centering
    \includegraphics[width=\linewidth]{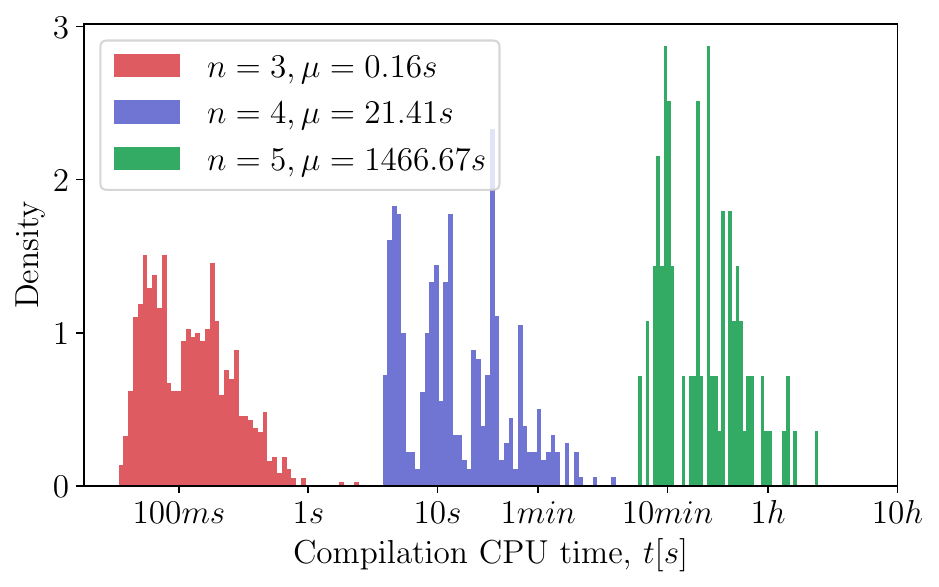}
    \caption{Total compilation CPU time for $1000$ ($500$, $100$) random target unitaries on $n=3$ ($n=4$, $n=5$) qubits. Variational optimizations are run with randomized initial parameters until the first run that achieves a cost below $\epsilon=10^{-10}$. We also report the average time to compilation $\mu$, but note the large variance between different targets.}
    \label{fig:compilation_times}
\end{figure}

\begin{figure*}
    \centering
    \includegraphics[width=\linewidth]{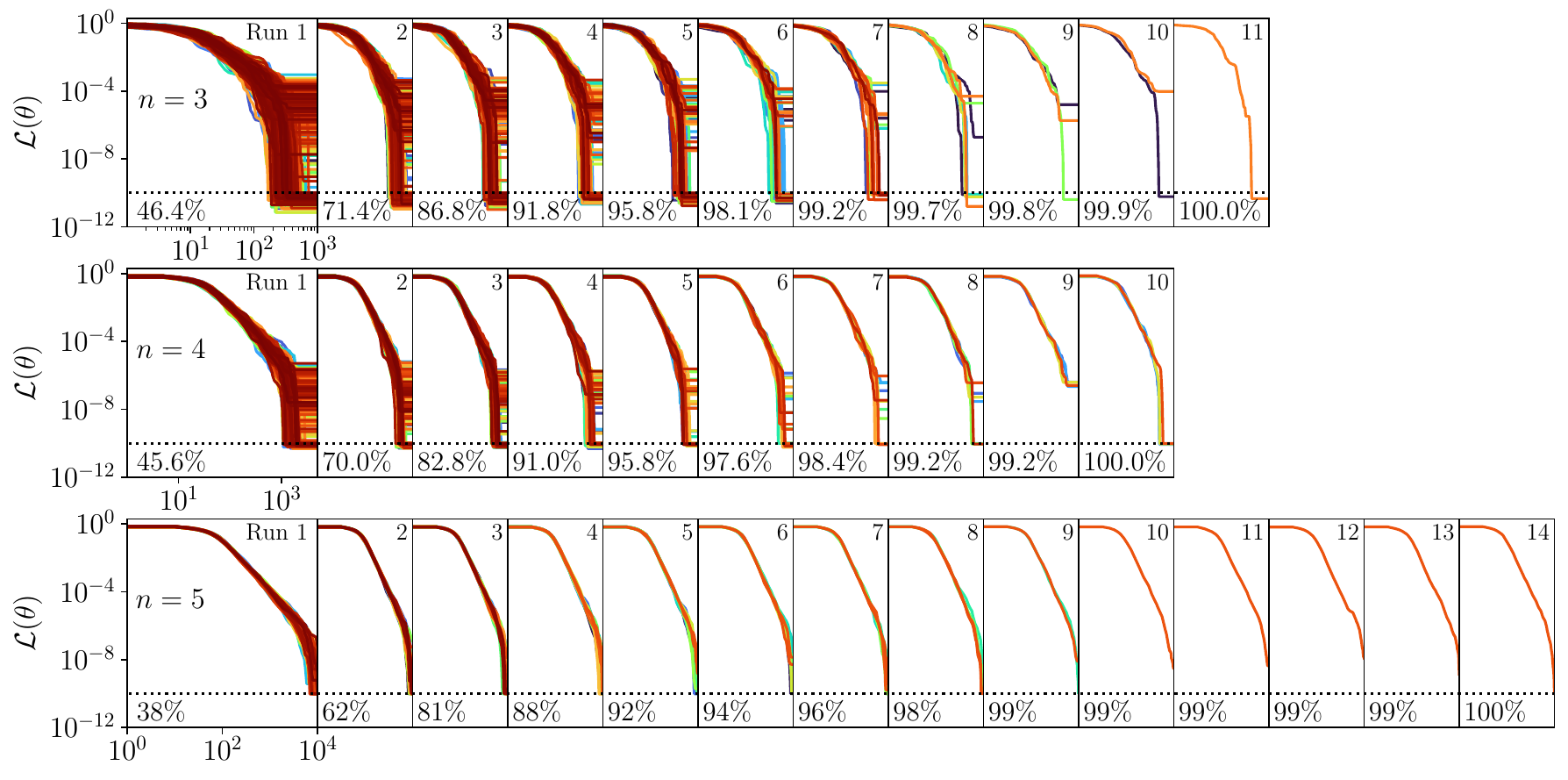}
    \caption{Cost during the variational optimization to compile $1000$ ($500$, $100$) random target unitaries for $n=3$, top ($n=4$, middle, $n=5$, bottom). For each target, randomly initialized optimizations are run until the first convergence to the target tolerance $\epsilon=10^{-10}$. The portion of converged runs is reported in the bottom of each panel. The axes for all panels per row match, the first panel is displayed wider for visibility purposes only.}
    \label{fig:compilation_costs}
\end{figure*}

\subsection{Applications}
\label{sec:applications}
The lower bounds for two-qubit gates and parameters are referring to typical representatives of $SU(2^n)$. There are, of course, unitary operators that can be realized by far fewer elementary gates. These unitaries form a measure zero subset of $SU(2^n)$ and are considered untypical. This terminology is slightly misleading, as most gates that we use to build quantum algorithms are untypical: $\text{CNOT}, \text{Toffoli}, R_X, R_Y, R_Z$ etc. are all untypical. The same goes for subroutines like quantum adders~\cite{Gidney2018halving} or phase gadgets~\cite{Cowtan2020}.

In all these scenarios, explicit decomposition strategies are known and usually better than anything we can achieve with generic circuit synthesis.
The generic circuit synthesis method \texttt{unicirc.compile} provided here is most useful in scenarios with $n=3, 4, 5$ qubits and unitary matrices that either are typical or are untypical, but have no known constructive decompositions.

In the case of such untypical unitaries where constructive decompositions are unknown or hard to obtain, we can use an adaptive version of our circuit compilation routine that builds up a partial version of the ansatz that iteratively adds bricks (red boxes) in \cref{fig:sun_circuit} until convergence is reached. We empirically observe that this will happen at the very latest when we reach the full circuit ansatz, which is in line with the conjectured universality. When compilation succeeds before reaching the full ansatz, we consider the matrix untypical. To account for local minima that the optimization can get stuck in, we may repeat the compilation at each stage of the adaptive compilation $N_\text{reps}$ times. This method is implemented as \pyth{unicirc.compile_adapt} in~\cite{our_repo}.

\subsubsection{Vibronic Hamiltonians}

Trotter-based time evolution greatly benefits from finding small diagonalizable fragments of the Hamiltonian, $H^\text{frag} = U D U^\dagger$, leading to terms of the form

\begin{equation}\label{eq:numerical_fast_forwarding}
    e^{-i t H^\text{frag}} = U e^{-i t D} U^\dagger.
\end{equation}
The diagonal part $e^{-i t D}$ can be readily and efficiently decomposed (see,~e.g.,~\cite{PennyLane-diagonal-unitary-decomp}), but the diagonalizing unitary $U$ still needs to be synthesized. Our method lets us decompose such unitaries with the minimal number of rotation gates and two-qubit gates (for typical unitaries). We explicitly compile the diagonalizing gates of the fragments of Anthracene (6 states, 3 qubits) and Pentacene (16 states, 4 qubits) investigated in~\cite{motlagh2025vibronicdynamics,XPrize} and report the resources in \cref{tab:vibronic}. We filter the resulting angles for values of $\tfrac\pi2 k$ with $k\in \Z$ (i.e., including zeros and Clifford angles), and report only the number of non-Clifford rotation angles. In the case of Anthracene there are only $6$ states that are embedded in a $n=3$ qubit unitary. Due to this incompleteness, we expect the resulting unitaries to be untypical, and indeed empirically find savings in the number of non-Clifford rotations compared to the full ansatz and the bound for typical unitaries (see \cref{tab:vibronic}). In the case of Pentacene, we use all $16$ states available to the $n=4$ qubit unitary and find that the majority of unitaries are close to being typical and require $240$ non-Clifford angles on average, saving only about $6\%$ compared to the bound of $255$.

\begin{table}
\centering
\begin{tabular}{l|c|c|c|c|c}
Molecule & $N_\text{states}$ & n & bound (typical) & full & adaptive  \\
\hline
Anthracene & 6 & 3 & 63 & $55 \pm 8$ & $41 \pm 13$  \\
Pentacene & 16 & 4 & 255 & $253 \pm 5$ & $240 \pm 15$
\end{tabular}
\caption{Guaranteed non-Clifford rotation count and empirical average using our synthesis tools (\texttt{compile}) and its adaptive version (\pyth{compile_adapt}). The savings in Anthracene are likely due to the fact that we are embedding $N_\text{states}=6$ states in $2^n = 8$ dimensions. The savings for Pentacene, where the $2^n=16$ states span the full Hilbert space, are more modest, as one would expect.} 
\label{tab:vibronic}
\end{table}

\subsubsection{Fast-forwardable Hamiltonians}

In the previous section we looked at small Hamiltonian fragments that can be numerically diagonalized, and then performed unitary synthesis on the diagonalizing unitaries in \cref{eq:numerical_fast_forwarding}. This procedure is called fast-forwarding a Hamiltonian (fragment). There is a different flavour of fast-forwarding using horizontal Cartan decompositions from Lie theory~\cite{Kokcu2022}. This requires finding a suitable horizontal Cartan decomposition, which is not always easy to perform.

As in \cref{fig:adapt-Ising}, we test our adaptive method on the Ising model and find that our method matches or beats the non-Clifford resources compared to fast-forwarding the Hamiltonian with a horizontal Cartan decomposition. We defer the explicit horizontal Cartan decomposition to \cref{sec:fast-forward}. The adaptive compiler mostly matches the horizontal Cartan decomposition and occasionally yields modest savings. Thus the main advantage of \pyth{compile_adapt} is the convenience of using it as a black-box compiler and not having to worry about finding a suitable horizontal Cartan decomposition, which can be a tedious task.

\begin{figure}
    \centering
    \includegraphics[width=0.9\linewidth]{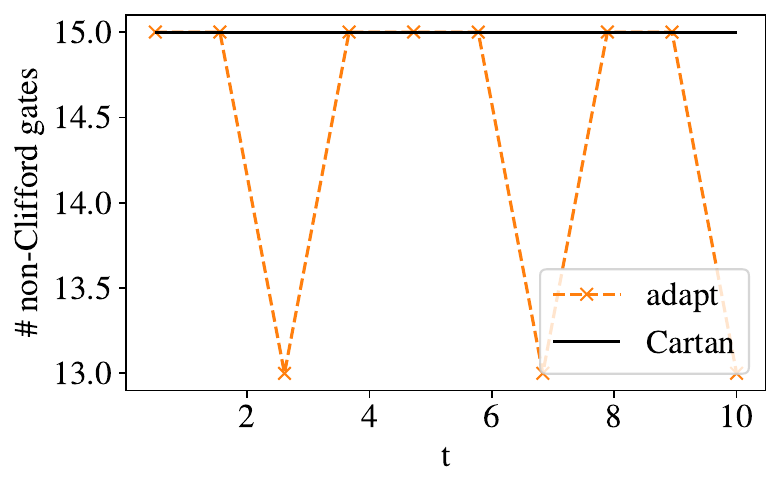}
    \caption{Comparing the number of non-Clifford rotation gates when we decompose the Ising Hamiltonian using \pyth{compile_adapt} (``adapt") and the horizontal Cartan decomposition (``Cartan"). \pyth{compile_adapt} matches or beats the horizontal Cartan decomposition without any physical insights on the Hamiltonian fragments.}
    \label{fig:adapt-Ising}
\end{figure}

\section{Mathematical considerations}\label{sec:math}

In this section we present some statements that we believe to be useful for a proof of universality of the presented circuit families. Unfortunately, we do not have a complete proof, because we are not able to prove \cref{conj:no_walls}. Additionally, due to numerical limitations, we only know for $n\leq 5$ that the Jacobian matrix of our circuits has full rank for at least one set of parameters.
We present the statements here in informal versions and defer more technical versions, together with proofs, to \cref{app:math_details}.

For simplicity, we think of parametrized quantum circuits (PQCs) as standard rotation gates $R_X$, $R_Y$, and $R_Z$, combined with static entangling gates like $\cnot$ or $\cz$, and we assume parameters to feed into the rotation gates without any further preprocessing. 
Due to periodicity of the rotation gates, such a PQC maps from a parameter space $T^d$, the $d$-dimensional torus, into some $d$-dimensional target group $\groupG$,
\begin{align}
    F:\ \ T^d &\to \groupG,\quad
    \theta \mapsto F(\theta)\in \C^{(2^n\times 2^n)},
\end{align}
where we express both the output of the PQC $F$ and the target group via complex matrices, using the fundamental representations of the classical groups $SU(2^n)$, $SO(2^n)$ and $Sp^\ast(2^n)$. Note that we already assumed the number of parameters to match the dimension $d$ of $\groupG$.
An illustrative consequence of our choice $T^d$ as parameter space is the following lemma:

\begin{restatable}{lemma}{densitysurjectivity}\label{lemma:density_surjectivity}
    If the image of a PQC is dense in $\groupG$, the image is the full group.
\end{restatable}
This Lemma tells us that the space of unitaries that an ansatz can reach is not punctured, which would break surjectivity. Instead, it implies that if a PQC is not surjective, it must fail to reach \textit{full-dimensional patches} of $\groupG$.
We can further illustrate the way this might happen by considering the Jacobian $J_F$ of the map $F$, and in particular its rank, which encodes the number of directions parametrized by $F$ on the tangent spaces of $\groupG$.
The \textit{regular set} of $F$ consists of the points in $T^d$ at which $J_F$ has full rank $d$,~i.e., those points around which $F$ can parametrize changes into all directions on $
\groupG$. The \textit{critical set} $\mathcal{C}$ is its complement and contains points at which the movement of $F(\theta)$ is restricted to a lower-dimensional submanifold; see \cref{fig:no_walls}a).
With this, we can state the following characterization:

\begin{restatable}{lemma}{fullrankanywhereae}\label{lemma:full_rank_anywhere_ae}
    If the Jacobian of a PQC has full rank anywhere, it has full rank almost everywhere. That is, its critical set $\mathcal{C}$ has measure zero.
\end{restatable}

Given our characterization above, this lemma tells us that a full-dimensional patch of the group $\groupG$ is accessible to $F(\theta)$ at almost every point, as soon as this is the case at \textit{any} position in parameter space.
From an algebraic perspective, this implies that if the dynamical Lie algebra (DLA) of the PQC is the full Lie algebra $\mathfrak{g}$ of $\groupG$, the PQC will parametrize universal dynamics almost everywhere on its domain, and there are no $d$-dimensional patches of the parameter space on which the circuit encodes lower-dimensional dynamics.
We numerically verified for $n\leq 5$ that our circuits for $SU(2^n)$, $SO(2^n)$ and $Sp^\ast(2^n)$ have full Jacobian rank somewhere in $T^d$~\cite{our_repo}.

\begin{figure}
    \centering
    \def\svgwidth{\linewidth}
    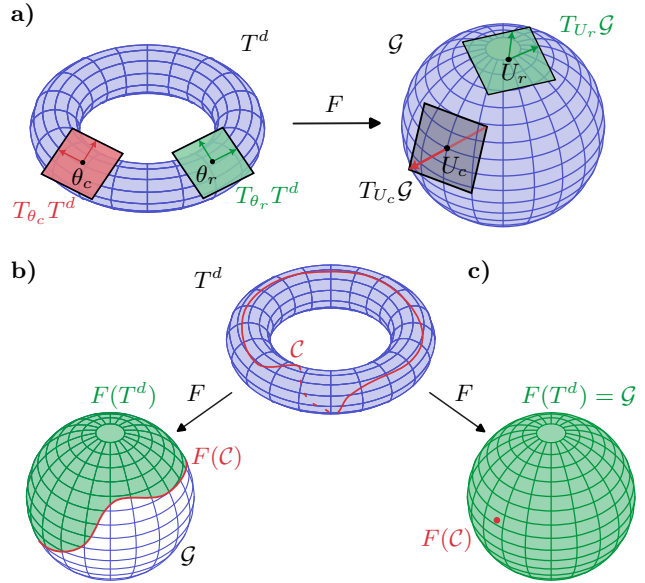
    \caption{Relevant geometrical scenarios for the surjectivity of a PQC matrix function $F$ from the torus $T^d$ into a $d$-dimensional Lie group $\groupG$, with $d=2$ in the figure. \textbf{a)} Regular points $\theta_r\in T^d$ of $F$ are such that its Jacobian $J_F(\theta_r)$ has (full) rank $d$, so that its image spans the entire tangent space $T_{U_r}\groupG$ at $U_r=F(\theta_r)$ (green). 
    By contrast, at critical points $\theta_c\in\mathcal{C}$, $J_F(\theta_c)$ is rank-deficient, so that its image is a proper subspace of $T_{U_c}\groupG$ at $U_c=F(\theta_c)$ (red).
    \Cref{lemma:full_rank_anywhere_ae} tells us that if there is a single regular point, almost all points in $T^d$ are regular points.
    \textbf{b)} The critical set $\mathcal{C}$ of $F$, while having measure zero in the case of \cref{lemma:full_rank_anywhere_ae}, might be $(d-1)$-dimensional and could be mapped to a $(d-1)$-dimensional submanifold $F(\mathcal{C})\subset\groupG$ that walls off a $d$-dimensional region of $\groupG$, preventing surjectivity of $F$. \Cref{conj:no_walls} states that this scenario actually does not occur for PQCs. \textbf{c)} Instead, it conjectures that $(d-1)$-dimensional critical sets are only ever mapped to at most $(d-2)$-dimensional subspaces of $\groupG$.
    }
    \label{fig:no_walls}
\end{figure}

\Cref{lemma:full_rank_anywhere_ae} also elucidates the way in which a PQC might fail to be surjective onto all of $\groupG$, which we illustrate in \cref{fig:no_walls}b); 
the critical set $\mathcal{C}\subset T^d$ is restricted to not have full dimension.
This implies that the critical values $F(\mathcal{C})\subset\groupG$ do not have full dimension either. 
However, they can form a $(d-1)$-dimensional boundary around the $d$-dimensional patch that can be reached by $F$, ``walling off" other $d$-dimensional regions of the target group $\groupG$ (\cref{fig:no_walls}b)).
One path to showing surjectivity of a PQC thus would be to show that the critical values actually form a subspace of $\groupG$ that is even lower-dimensional, namely at most $(d-2)$-dimensional. This would prevent the critical values from separating the image of the PQC, $F(T^d)$, from full-dimensional subspaces of the target group (\cref{fig:no_walls}c)).
Based on numerical experiments and lower-dimensional examples, we do conjecture this to be true:

\begin{conjecture}\label{conj:no_walls}
    The critical values $F(\mathcal{C})$ of a PQC that has full rank almost everywhere in $T^d$ form an at most $(d-2)$-dimensional subset of $\groupG$.
\end{conjecture}

If this statement holds true, the PQC in question will always be able to ``move around" the critical values $F(\mathcal{C})$, which in turn will no longer separate the PQC image from full-dimensional subspaces of $\groupG$, asserting global surjectivity.
To summarize, if \cref{conj:no_walls} holds, we know our circuit families to be universal for $n\leq 5$.  
For $n>6$, we additionally require a proof that the Jacobian of the PQCs has full rank for some $\theta\in T^d$, implying full rank almost everywhere via \cref{lemma:full_rank_anywhere_ae}.

Useful properties for a proof of \cref{conj:no_walls} might hide in the algebraic geometry underlying the ansatz. Additional helpful features could lie in the analyticity of the circuits as well as the orthonormal Pauli basis of the Lie algebra $\mathfrak{g}$ that appears in the generators of the rotation gates. In this context, the rephrasing of our ansatz as a  PPR ansatz in \cref{app:pprs} might come in handy.

\section{Conclusion}

We presented families of brick wall circuits with optimal rotation gate and two-qubit gate counts to parametrize $SU(2^n)$, $SO(2^n)$, and $Sp^\ast(2^n)$. We conjecture these circuits to be universal for their respective group, and provide numerical evidence supporting this conjecture for $n=3, 4, 5$.
For instance, we compiled random unitaries to the proposed circuits, demonstrating that they can be used for unitary synthesis in practice. In addition, we showed numerically that their Jacobian has full rank at at least one point in parameter space, and proved that this implies full rank almost everywhere in parameter space.
While this does not rigorously prove universality, we stated a technical conjecture that would imply universality from observing full Jacobian rank anywhere in parameter space.

Currently, compilation is performed via variational optimization, which can be slow in practice even when using special-purpose numerical methods and matured tools like JAX and Optax. It would bring a great improvement in compilation time to find linear-algebra-based techniques to directly infer the rotation angles from the unitary matrix, as is done in methods based on recursive Cartan decompositions~\cite{Shende-Bullock-Markov,wierichs2025recursive}. Solving the system of non-linear equations $U_\text{target} = U_\text{ansatz}(\theta)$ numerically or analytically for the parameters $\theta$ using computer-assisted methods did not prove fruitful as it was already intractable for $n=3$.
Another interesting next step will be to understand the observed difficulty in compiling $SO(8)$ matrices.

\section{Acknowledgements}
We thank Paarth Jain, Robert Lang, and Danial Motlagh for providing data for vibronic Hamiltonian fragments, and Max West for providing code for sampling Haar random symplectic unitary matrices.

\bibliographystyle{quantum}
\bibliography{main}

\appendix

\section{Circuits for \texorpdfstring{$SO(2^n)$}{SO(2\^{}n)}}
\label{app:SON}

\begin{table*}[ht]
\centering
\begin{tabular}{l|c|c|c}
$n$ & $3$ & $4$ & $5$ \\
\hline
circuit & 
\scalebox{0.7}{
\begin{quantikz}
    & \gate{Y} & \ctrl{1}\block\gategroup[3,steps=4,style={inner sep=3pt,rounded corners,color=WeakGreen}]{$\times 6$} & \gate{Y} & & & \ctrl{1}\gategroup[3,steps=2,style={inner sep=3pt,rounded corners,color=XanaBlue}]{remainder} & \gate{Y} & \\
    & \gate{Y} & \ctrl{0} & \gate{Y} & \ctrl{1}\block & \gate{Y} & \ctrl{0} & &\\
    & \gate{Y} & & & \ctrl{0} & \gate{Y} & & &
\end{quantikz}
}

&

\scalebox{0.7}{
\begin{quantikz}
    & \gate{Y} & \ctrl{1}\block\gategroup[4,steps=4,style={inner sep=3pt,rounded corners,color=WeakGreen}]{$\times 19$} & \gate{Y} & & & \ctrl{1}\block\gategroup[4,steps=2,style={inner sep=3pt,rounded corners,color=XanaBlue}]{remainder} & \gate{Y} & \\
    & \gate{Y} & \ctrl{0} & \gate{Y} & \ctrl{1}\block & \gate{Y} & \ctrl{0} & \gate{Y} & \\
    & \gate{Y} & \ctrl{1}\block & \gate{Y} & \ctrl{0} & \gate{Y} & & & \\
    & \gate{Y} & \ctrl{0} & \gate{Y} & & & & &
\end{quantikz}
}

&

\scalebox{0.7}{
\begin{quantikz}
    & \gate{Y} & \ctrl{1}\block\gategroup[5,steps=4,style={inner sep=3pt,rounded corners,color=WeakGreen}]{$\times 61$} & \gate{Y} & & & \ctrl{1}\block\gategroup[5,steps=2,style={inner sep=3pt,rounded corners,color=XanaBlue}]{remainder} & \gate{Y} & \\
    & \gate{Y} & \ctrl{0} & \gate{Y} & \ctrl{1}\block & \gate{Y} & \ctrl{0} & \gate{Y} & \\
    & \gate{Y} & \ctrl{1}\block & \gate{Y} & \ctrl{0} & \gate{Y} & \ctrl{1} & \gate{Y} & \\
    & \gate{Y} & \ctrl{0} & \gate{Y} & \ctrl{1}\block & \gate{Y} & \ctrl{0} & & \\
    & \gate{Y} & & & \ctrl{0} & \gate{Y} & & &
\end{quantikz}
}

\\
\# parameters & $3 + 6\cdot2\cdot2+1 = 28$  & $4 + 19\cdot3\cdot2 + 2 = 120$ & $5 + 61 \cdot 4 \cdot 2 + 3 = 496$ \\
\# $\cz$s & $6\cdot2+1 = 13$ & $19\cdot3 + 1 = 58$ & $61\cdot4 + 2 = 246$ \\
\end{tabular}
\caption{Universal $SO(2^n)$ circuits for $n=3, 4, 5$ with optimal parameter count ($\text{dim}_{SO(2^n)} = 2^{n-1} (2^n - 1)$) and number of $\cz$ gates $\left\lceil\tfrac{1}{2}(\text{dim}_{SO(2^n)} - n)\right\rceil$ (see \cref{eq:lower_bound_2q_general}).}
\label{tab:SON}
\end{table*}

The optimal $SO(2^n)$ templates for $n=3, 4, 5$ are given in \cref{tab:SON}, following the pattern shown in \cref{sec:main_results}. Their universality is confirmed by the rank test of the Jacobian matrix, equalling the group's dimension; see~\cite{our_repo} for further detail. The minimum required two-qubit gate count can be derived in the same fashion as for $SU(2^n)$. We can always start with a layer of single-qubit $SO(2)$ gates on each of the $n$ qubits, which corresponds to just $R_Y$ rotations, implying $n_\text{initial}=n$. Then we can include $2$ parameterized $R_Y$ gates per $\cz$ that we add,~i.e., $n_\text{params/2q}=2$. Further, the dimension of $SO(2^n)$ is $d=\tfrac12 2^n (2^n -1)$, so that we obtain the lower bound also shown in \cref{eq:lower_bound_2q_so}:

\begin{align}
    c &\geq \left\lceil \frac{1}{2} \left(\frac{1}{2} 2^n (2^n -1) - n\right)\right\rceil.
\end{align}

\begin{figure}
    \centering
    \includegraphics[width=\linewidth]{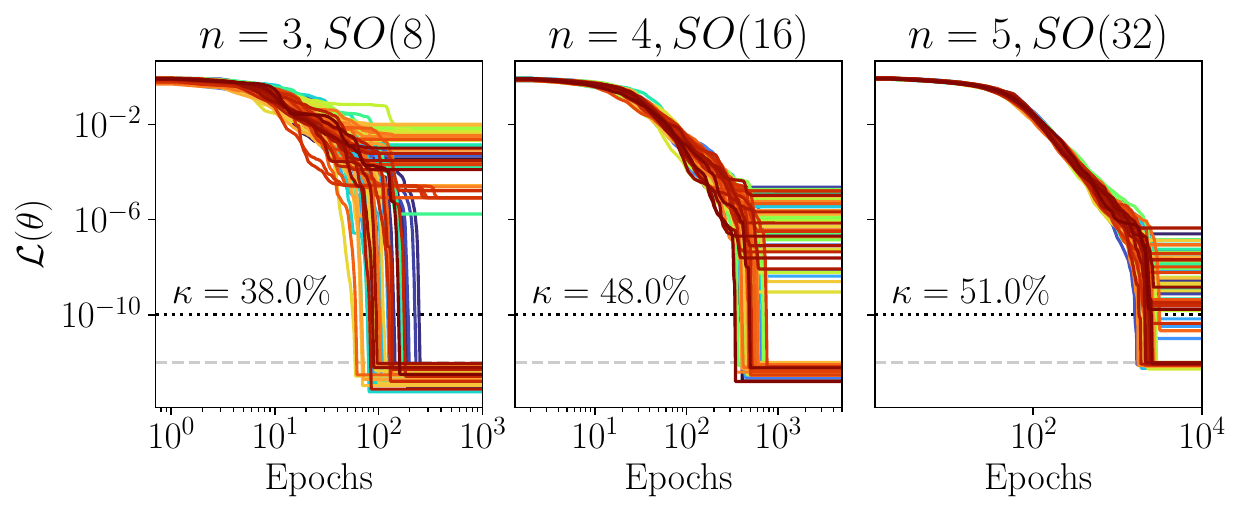}
    \includegraphics[width=\linewidth]{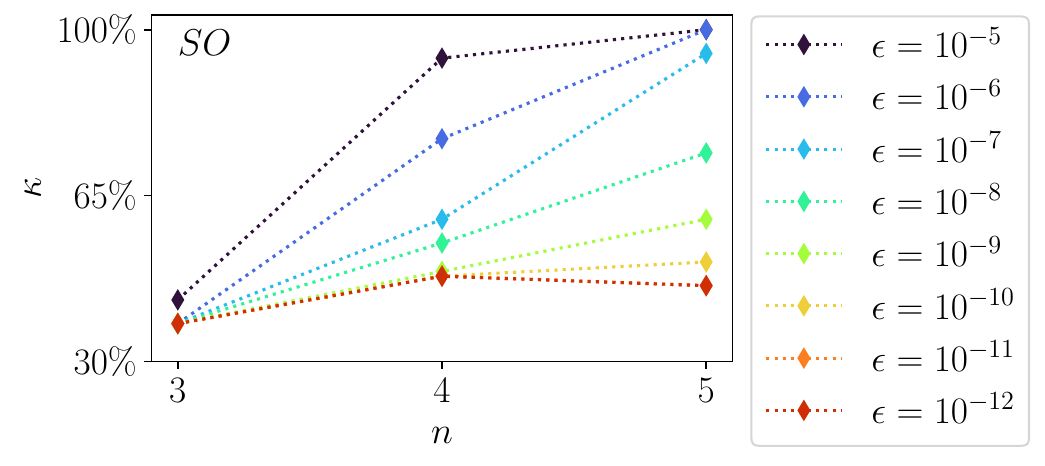}
    \caption{Cost function $\mathcal{L}(\theta)$ and success rates $\kappa$ for compiling $10$ random special orthogonal matrices on $n=3,4,5$ qubits, using $10$ attempts each; see \cref{fig:convergence,fig:successprob} and the text for more details. Note that for $n=3$,~i.e., the group $SO(8)$, three targets were not compiled successfully for any of the attempts.
    }
    \label{fig:convergence_success_rates_so}
\end{figure}

In addition to the Jacobian rank test, which is passed by the shown circuits for $n=3,4,5$, we perform the same numerical compilation experiment as in \cref{fig:convergence,fig:successprob} and report the results in \cref{fig:convergence_success_rates_so}. We find reliable compilation results, with the exception of some random $SO(8)$ matrices.
Our compiler did not succeed in finding circuit parameters for these ``difficult targets", even when allowing for more randomly initialized attempts.
Therefore, the evidence for universality for $SO(8)$ is less conclusive than for the other groups.

\section{Circuits for \texorpdfstring{$Sp^\ast(2^n)$}{Sp*(2\^{}n)}}
\label{app:SpN}

\begin{table*}[ht]
\centering
\begin{tabular}{l|c|c|c}
$n$ & $3$ & $4$ & $5$ \\
\hline
circuit & 

\scalebox{0.6}{
\begin{quantikz}
& \gate{ZYZ} & \gate{iY}\block \gategroup[3,steps=3,style={inner sep=3pt,rounded corners,color=WeakGreen}]{$\times 5$} & \gate{ZY} & \gate[3]{\ \ \ }\gategroup[3,steps=1,style={inner sep=0pt,rounded corners, color=Red,fill=white}]{}& \gate{iY}\gategroup[3,steps=2,style={inner sep=3pt,rounded corners,color=XanaBlue}]{remainder} & \gate{Z} & \\
& \gate{Y} & \ctrl{-1} & \gate{Y} & & \ctrl{-1} & & \\
& \gate{Y}& & & & & &
\end{quantikz}%
}

&

\scalebox{0.6}{
\begin{quantikz}
& \gate{ZYZ} & \emptyblock \gategroup[4,steps=3,style={inner sep=3pt,rounded corners,color=WeakGreen}]{$\times 14$} & \gate[3]{\ \ \ }\gategroup[3,steps=1,style={inner sep=0pt,rounded corners, color=Red,fill=white}]{} & \gate[4]{\ \ \ }\gategroup[4,steps=1,style={inner sep=0pt,rounded corners, color=Red,fill=white}]{} & \emptyblock\gategroup[4,steps=3,style={inner sep=3pt,rounded corners,color=XanaBlue}]{remainder} & \gate{iY} & \gate{Z} & \\
& \gate{Y} & & & & & & & \\
& \gate{Y} & & & & & \ctrl{-2}& & \\
& \gate{Y} & & & & & & &
\end{quantikz}
}

&

\scalebox{0.6}{
\begin{quantikz}
& \gate{ZYZ} & \gate[5]{\qquad} \gategroup[5,steps=1,style={inner sep=3pt,rounded corners,color=WeakGreen,fill=white}]{$\times 43$} & \gate{iY}\block\gategroup[5,steps=4,style={inner sep=3pt,rounded corners,color=XanaBlue}]{remainder} & \gate{ZY} & \gate{iY} & \gate{Z Y} & \\
& \gate{Y} & & \ctrl{-1} & \gate{Y} & & &\\
& \gate{Y} & & & &\ctrl{-2} & &\\
& \gate{Y} & & & & & & \\
& \gate{Y} & & & & & &
\end{quantikz}%
}

\\
\# parameters & $5 + 5\cdot6 + 1 = 36$  & $6 + 14\cdot9 + 4 = 136$ & $7 + 43 \cdot 12 + 5 = 528$ \\
\# $\text{C}(i\text{Y}$)s & $5\cdot2 + 1 = 11$ & $14\cdot3 + 2 = 44$ & $43\cdot4 + 2 = 174$ \\
\end{tabular}
\caption{Universal $Sp^\ast(2^n)$ circuits for $n=3, 4, 5$ with optimal parameter count ($\text{dim}_{Sp^\ast(2^n)} = 2^{n-1} (2^n + 1)$) and number of $\ciy$ gates $\left\lceil\tfrac{1}{3}(\text{dim}_{Sp^\ast(2^n)} - ( n + 2)\right\rceil$ (see \cref{eq:lower_bound_2q_general}). Red blocks extending to more than two qubits are two-qubit blocks acting on the first and last qubit in the block.}
\label{tab:SpN}
\end{table*}

The universal $Sp^\ast(2^n)$ circuits are summarized in \cref{tab:SpN}. Their universality is confirmed by the rank test of the Jacobian matrix, equalling the group's dimension~\cite{our_repo}. Note that even though $Sp^\ast(2^n)$ is higher dimensional than $SO(2^n)$, the circuits have fewer two-qubit gates. That is because $Sp^\ast(2^n)$ can be packed more densely with $3$ parameters per two-qubit gate (while maintaining universality), relative to the group's dimension (see \cref{eq:N_two_q_symplectic}).

The symplectic group\footnote{We use the physicist's convention of calling the symplectic unitary group simply the symplectic group. Further we call $Sp(\tfrac12 2^n) =: Sp^\ast(2^n)$.} is different to $SU(2^n)$ and $SO(2^n)$ in that there is one special qubit that acts differently than the others. This can be seen from the symplectic condition 

\begin{equation}
    Q \in Sp^\ast(2^n) \Leftrightarrow Q J Q^T = J,
\end{equation}
where $J$ is a fixed symplectic form, e.g., $J = i Y \otimes \id_{n-1}$\footnote{This choice of $J$ is arbitrary and we could choose a different qubit for the $Y$ operator. $\id_{n-1}$ is the identity matrix on $n-1$ qubits.}.

To build the symplectic circuit ansatz, we look at the symplectic algebra $\mathfrak{sp^\ast}(2^n)$ that generates the symplectic group. It can be parametrized as

\begin{equation*}
    \mathfrak{sp^\ast}(2^n) = \text{span}_\R \{\id \otimes \mathfrak{so}(2^{n-1}), \{X, Y, Z\} \otimes \mathfrak{so}(2^{n-1})^\perp \},
\end{equation*}
and has dimension $d=\tfrac12 2^n (2^n + 1)$. From this we see that the special qubit allows for arbitrary single-qubit gates, whereas all other qubits only allow for $R_Y$ rotations. For that reason, the initial layer of the circuit has $n_\text{initial}=n+2$ parameters, and the tightest possible packing of a symplectic circuit is achieved by maximizing the usage of the special qubit, as shown in \cref{sec:main_results}.

We used $\ciy$ as the entangler targeting the special qubit to ensure every gate in the ansatz is symplectic. Compared to a $\cy$ gate, it has an additional phase-shift $S$ gate on the control qubit:

\begin{center}
\scalebox{0.8}{
\begin{quantikz}
& \gate{iY} & 
\midstick[2,brackets=none]{=}
& \gate{Y} & & 
\midstick[2,brackets=none]{=}
& & \gate{Y} & \\
& \ctrl{-1}\wire[u]{q} &
& \ctrl{-1} & \gate{S} &
& \gate{S} & \ctrl{-1} &
\end{quantikz}%
}
\end{center}
The number of parameters per two-qubit gate in the ansatz is $n_\text{params/2q}=3$. Thus, the lower bound is given by

\begin{equation} \label{eq:N_two_q_symplectic}
    c \geq \left \lceil\frac{1}{3} \left(\frac{1}{2} 2^n (2^n + 1) - (n +  2)\right)\right\rceil
\end{equation}

\begin{figure}
    \centering
    \includegraphics[width=\linewidth]{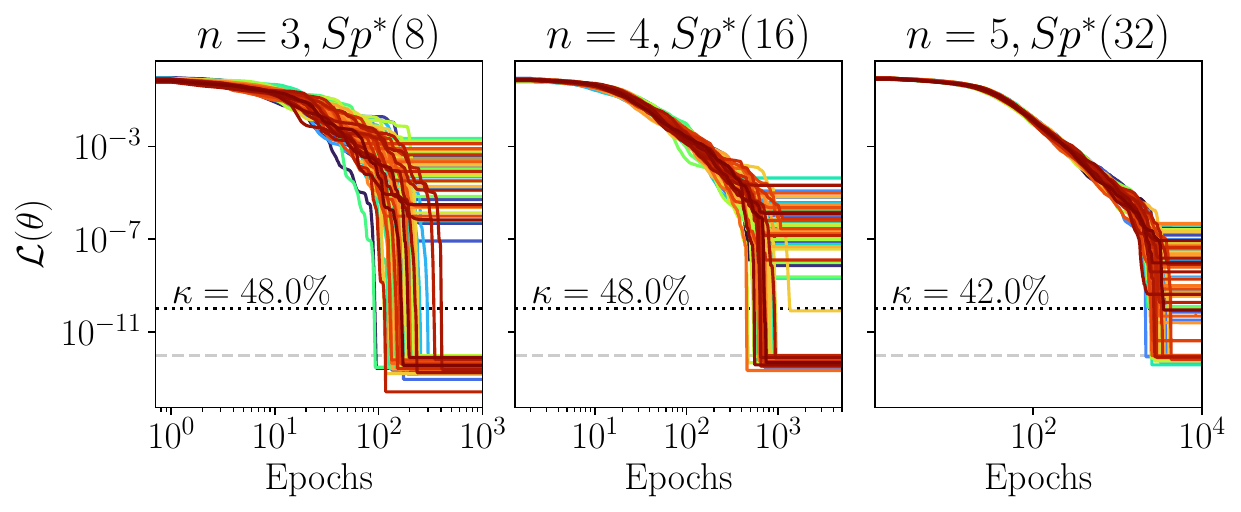}
    \includegraphics[width=\linewidth]{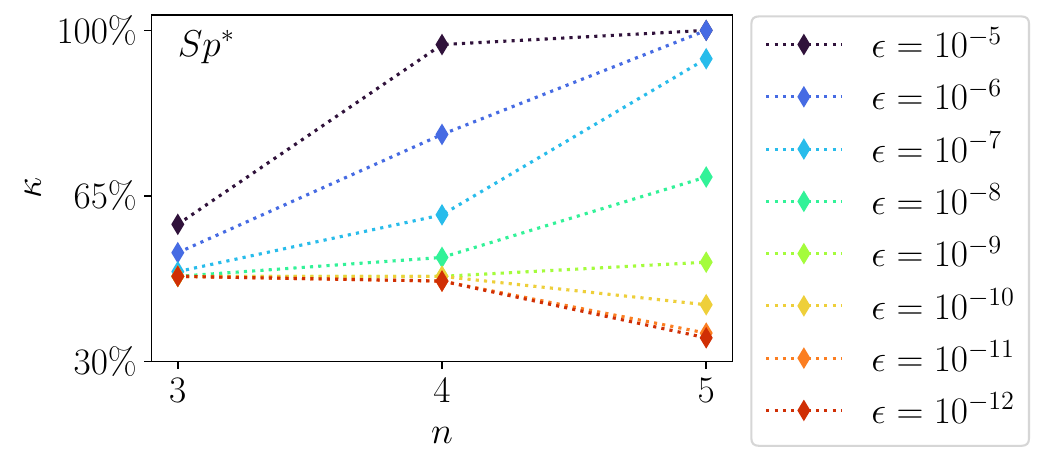}
    \caption{Cost function $\mathcal{L}(\theta)$ and success rates $\kappa$ for compiling $10$ random symplectic matrices on $n=3,4,5$ qubits, using $10$ attempts each; see \cref{fig:convergence,fig:successprob} and the appendix text for more details. 
    }
    \label{fig:convergence_success_rates_sp}
\end{figure}

In addition to the Jacobian rank test, which is passed by the shown circuits for $n=3,4,5$, we again perform the same numerical compilation experiment as in \cref{fig:convergence,fig:successprob} and report the results in \cref{fig:convergence_success_rates_sp}. We find reliable compilation results for all tested instances.

\section{Pauli product rotations (PPRs)}\label{app:pprs}

Most surface-code-based QEC schemes target Pauli product rotations (PPRs) as a gate set~\cite{Litinski2019}. It can therefore be desirable to directly provide circuits in this format. We directly translate our universal circuits in \cref{tab:SUN,tab:SON,tab:SpN}. We first note that most gates are already (single-qubit) Pauli rotation gates. All other static two-qubit gates (i.e., the $\cz$ and $\ciy$ gates) can be cancelled by using the commutation rules provided in Fig. 4 in~\cite{Litinski2019} and using the fact that $\cz^2 = \id$. For example, the $SU(2^3)$ circuit has the following regular form containing only PPRs:

\begin{center}
\scalebox{0.62}{
\begin{quantikz}
& \gate{ZYZ} & \gate[2,disable auto height]{\verticaltext{Y Z}} \gategroup[3,steps=10,style={inner sep=3pt,rounded corners,color=WeakGreen}]{$\times 4$} \gategroup[3,steps=5,style={inner sep=0pt,rounded corners,color=Red}]{} & \gate[2,disable auto height]{\verticaltext{Z Y}} & \gate{Z} & \gate[3,disable auto height]{\verticaltext{Z Y Z}} & & & \gate{Y} & \gate{Z} & & & \\
& \gate{ZYZ} &  & & \gate{Z} & & \gate[2,disable auto height]{\verticaltext{Z Y}} & \gate{Z} & \gate[2,disable auto height]{\verticaltext{Y Z}} & \gate{Z} & \gate{Y} & \gate{Z} & \\
& \gate{ZYZ} &  & & & & & \gate{Z} &  & & \gate{Y} & \gate{Z} &
\end{quantikz}%
}
\end{center}

The green block is repeated four times, though in the last block it only contains the gates in the inner red block. This way we obtain the desired $9 + 3*16 + 6 = 63$ parameters. This form was obtained by simply taking the $SU(2^3)$ circuit and commuting the static $\cz$ gates to cancel each other.
The parameters map 1-to-1 from the PPR and $(CZ, R_Y, R_Z)$ version of the circuit. Since the circuit structure is identical, we do not expect any significant changes in convergence when compiling using the PPR structure.

\subsection{Lie algebra product ansatz}

\begin{figure}
    \centering
    \includegraphics[width=1\linewidth]{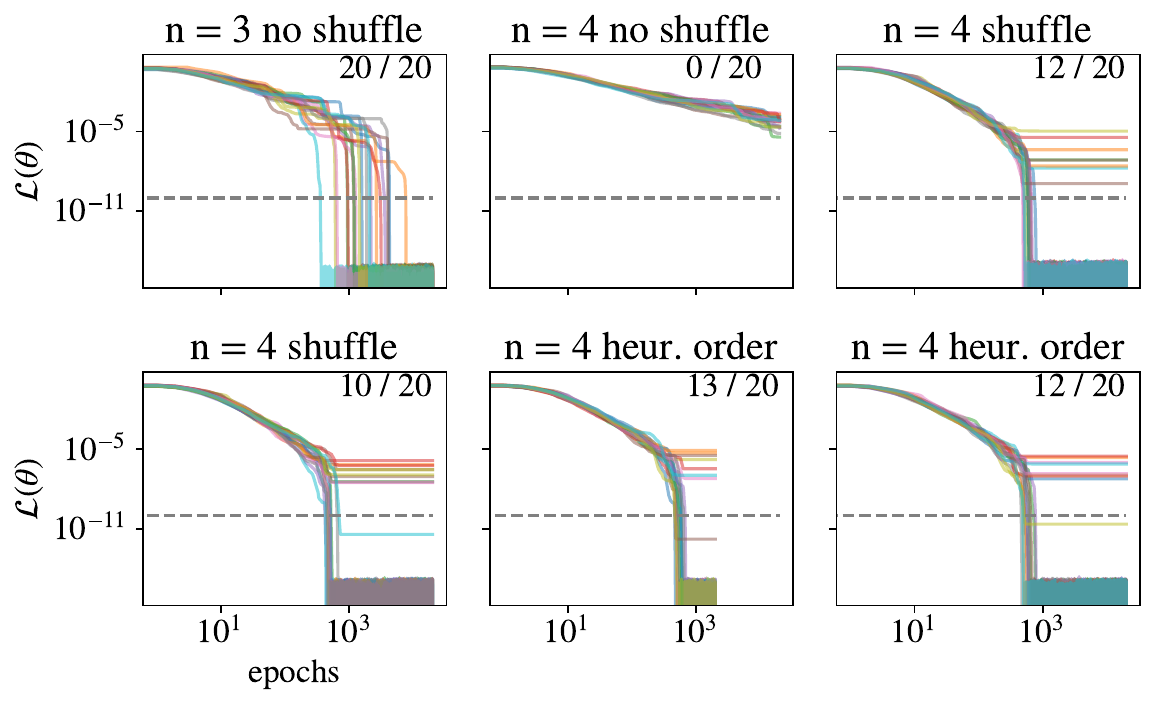}
    \caption{Cost function during optimization of simple Lie-algebra-product circuits, \cref{eq:lie_algebra_product_ansatz}. For $n=3$, a lexicographical ordering of Pauli words is used and found to work exceptionally well, with a $100\%$ convergence rate. For $n=4$, the same is hard to train, but becomes trainable when we shuffle the ordering, as exemplified by two random shufflings. A heuristic ordering using a greedy algorithm to order the Pauli operators such that no consecutive operators commute also yields good results. We show two examples, starting from the first and last word of the lexicographical ordering, respectively. See~\cite{our_repo} for more details.}
    \label{fig:PPR_full_algebra}
\end{figure}

We note that one can also construct universal circuits using the Lie algebra generators. Let us take the Pauli words $P$ as the generators of $\mathfrak{su}(2^n)$. Then we know that any unitary matrix $U\in SU(2^n)$ can be written as

\begin{equation}
    U = \exp\left(-i \sum_{P \in \mathfrak{su}(2^n)} \theta_P P\right).
\end{equation}
We claim that a product ansatz in terms of PPRs
\begin{equation}\label{eq:lie_algebra_product_ansatz}
    U = \prod_{P \in \mathfrak{su}(2^n)} e^{-i \theta_P P},
\end{equation}
with each of the Pauli words appearing only once, is also universal as long as we are careful with their ordering. In particular, we found that a lexicographical ordering of the Pauli words does not train well for $n\geq 4$, as seen in \cref{fig:PPR_full_algebra}. 

There are two different ways of making the ansatz universal. On one hand, we can randomly shuffle the Pauli words from the lexicographical ordering. In the unlikely event that the ansatz does not pass the rank test from \cref{sec:derivation_rank_test}, we shuffle again until that is the case. Alternatively, we can do a simple heuristic ordering: starting from a random or fixed operator, we order them such that no two consecutive Pauli words commute. For the latter, a simple greedy algorithm works (see ~\cite{our_repo}). In \cref{fig:PPR_full_algebra}, we show convergence for an ansatz with $n=4$ that has been randomly shuffled and using the heuristic odering. We tried different versions for $n=4$ and found that they all have convergence rates of around $50\%$ (see~\cite{our_repo}) or better.
One curiosity is that the circuit with lexicographical ordering works exceptionally well for $n=3$ and has a $100\%$ convergence rate as indicated in \cref{fig:PPR_full_algebra}.

\section{Expressibility by Sim et al. for \texorpdfstring{$n=4, 5$}{n=4,5}}\label{app:expressibility}

\begin{figure*}
    \centering
    \includegraphics[width=0.49\linewidth]{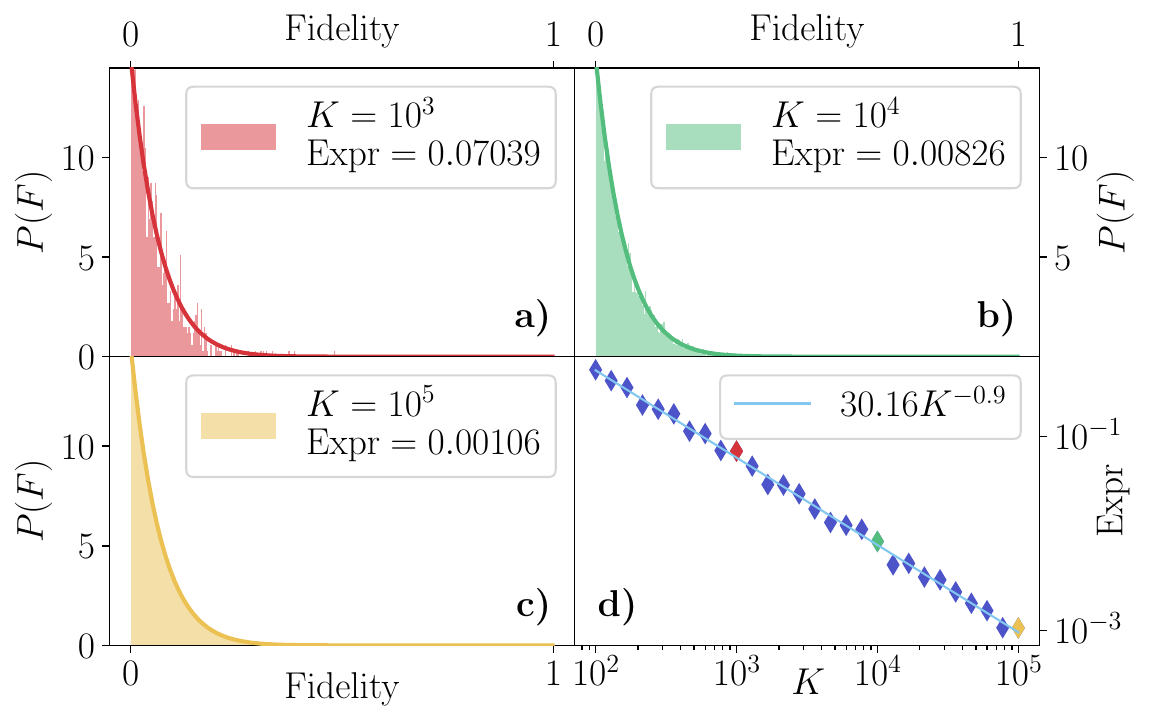}
    \includegraphics[width=0.49\linewidth]{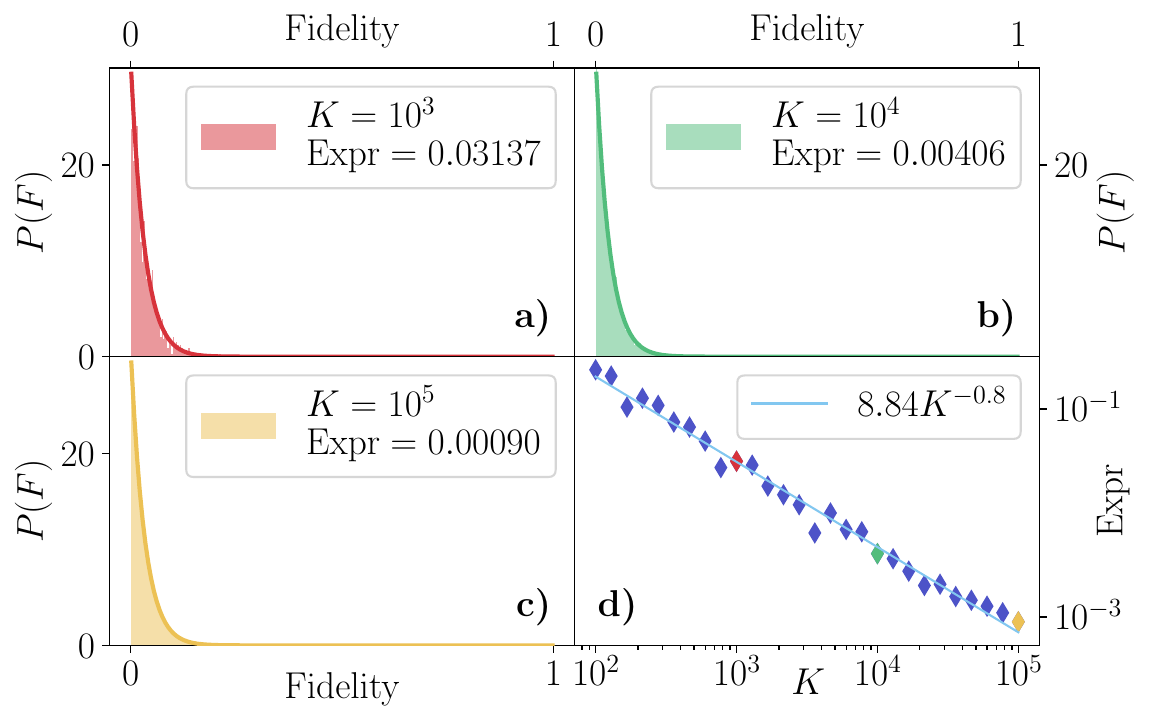}
    \caption{Numerical expressibility test from~\cite{Sim-Johnson-Aspuru-Guzik} for the regular, parameter-optimal circuit on four and five qubits. We repeat the same experiment as in \cref{fig:expressibility} for $n=4$ (left) and $n=5$ (right), restricting ourselves to smaller sample sizes for convenience. The same trend as for $n=3$ is observed.}
    \label{fig:expressibility_app}
\end{figure*}

Here we show the numerical results for the expressibility test by Sim et al.~\cite{Sim-Johnson-Aspuru-Guzik} for $n=4$ and $n=5$ for a smaller range of samples $100\leq K \leq 10^5$; see \cref{fig:expressibility_app}. We observe the same convergence as for $n=3$ in \cref{fig:expressibility}.

\section{Application to fast-forwardable Hamiltonians}
\label{sec:fast-forward}

A Hamiltonian with a polynomially sized \href{https://pennylane.ai/qml/demos/tutorial_liealgebra}{dynamical Lie algebra}~\cite{Korbinian2024lie} is said to be fast-forwardable, when there is a non-trivial \href{https://pennylane.ai/qml/demos/tutorial_fixed_depth_hamiltonian_simulation_via_cartan_decomposition#goal-fast-forwarding-time-evolutions-using-the-kak-decomposition}{horizontal Cartan decomposition}\footnote{The word \textit{horizontal} here refers to the special case where the Hamiltonian $H$ is in the horizontal subspace of the Cartan decomposition, such that the KAK decomposition is of the form $H = K A K^\dagger$, rather than the more general \href{https://pennylane.ai/qml/demos/tutorial_kak_decomposition}{KAK decomposition}~\cite{DavidWierichs2024KAK} form $H = K_1 A K_2$, which does not yield a fast-forwarding.}~\cite{Kokcu2022,Kottmann2024demo}.
Such Hamiltonians are also interesting targets for our adaptive compilation scheme. Let us look at the Ising Hamiltonian,

\begin{equation}\label{eq:ising_Ham}
    H_\text{Ising} = \sum_{j=1}^n Z_j + \sum_{j=1}^{n-1} X_j X_{j+1},
\end{equation}
which we know how to efficiently fast-forward using a horizontal KAK decomposition.
We can use the techniques described in~\cite{Kokcu2022} with the implementation from~\cite{Kottmann2024demo}, or the non-variational alternative provided in~\cite{wierichs2025recursive}, to find the horizontal KAK decomposition of $H_\text{Ising}$. This then yields

\begin{equation}\label{eq:KAK_Ising}
    H_\text{Ising} = K_\text{Ising} A_\text{Ising}  K_\text{Ising}^\dagger,
\end{equation}
where $K_\text{Ising} = \prod_{j=1}^{|\mathfrak{k}|} e^{-i \theta_j k_j}$ and $A_\text{Ising} = \prod_{j=1}^{|\mathfrak{a}|} e^{-i \theta_j a_j}$.

The operators $k_j \in \mathfrak{k}$ and $a_j \in \mathfrak{a}$ stem from the Cartan decomposition of the dynamical Lie algebra, $\mathfrak{g} = \mathfrak{k} \oplus \mathfrak{p}$, where $\mathfrak{a}\subset\mathfrak{p}$ is a (horizontal) Cartan subalgebra within the horizontal space $\mathfrak{p}$. For $n=3$ and the concurrence involution $\Theta(g) = -g^T$ (see~\cite{Dagli2008} and \href{https://docs.pennylane.ai/en/stable/code/api/pennylane.liealg.concurrence_involution.html}{qml.liealg.concurrence\_involution}~\cite{concurrence_involution}), 
we obtain the dynamical Lie algebra

\begin{multline*}
    \mathfrak{g} = \{XXI, IXX, ZII, IZI, IIZ, YXI, XYI, \\ IXY, IYX, YYI, IYY, XZX, YZX, XZY, YZY\},
\end{multline*}
its subalgebra
\begin{equation*}
    \mathfrak{k} = \{k_j\} = \{YXI, XYI, IXY, IYX, YZX, XZY\},
\end{equation*} 
and one of its \href{https://docs.pennylane.ai/en/stable/code/api/pennylane.liealg.horizontal_cartan_subalgebra.html}{horizontal Cartan subalgebras} ~\cite{horizontalCSA}
\begin{equation}
    \mathfrak{a} = \{ a_j\} = \{XXI, YYI, IIZ\}.
\end{equation}

We now provide the resulting coefficients $\theta_j$ and corresponding $k_j$ operators for $K_\text{Ising}$:

\begin{align*}
-1.0387190968491462 \cdot YXI \\
 1.317475393343197 \cdot XYI \\
 -1.9850892304110668 \cdot IYX \\
 -0.4142929036161686 \cdot IXY \\
 0.25332093345170165 \cdot YZX \\
 1.038719096849147 \cdot XZY.
\end{align*}

Further, $\theta_j$ and the corresponding $a_j$ operators for $A_\text{Ising}$ are as follows:

\begin{align*}
 1.2469796037174674 \cdot XXI \\
 1.8019377358048387 \cdot IIZ \\
-0.44504186791262884 \cdot YYI.
\end{align*}

We refer to \cite{our_repo} for the explicit computation. All coefficients are non-Clifford, so we overall have $2\cdot 6 + 3 = 15$ non-Clifford \href{https://pennylane.ai/compilation/pauli-product-rotations}{Pauli product rotations}~\cite{PPR}.

Our black-box adaptive compilation strategy \texttt{unicirc.compile\_adapt} matches (or improves) the $15$ non-Clifford rotation gates of the KAK decomposition without any encoded knowledge over a period of times from $t=[0.5, 10]$, as indicated in \cref{fig:adapt-Ising}.

\section{Details and proofs for mathematical considerations}\label{app:math_details}
Here we provide a more formal version of and proofs for the mathematical statements in \cref{sec:math}.
First, we need to make our notion of a PQC more rigorous. We assume knowledge of the Clifford and Pauli groups.

\begin{definition}
    A parametrized quantum circuit (PQC) is a quantum circuit consisting of Clifford gates and $d$ rotations chosen from the set $\{R_X, R_Y, R_Z\}$, with one value of the input vector feeding into one rotation each, without preprocessing. We identify it with the map
    \begin{align*}
        F : T^d\to\groupG\subset\C^{2^n\times 2^n},
    \end{align*}
    that results from evaluating the matrix of the circuit in some basis.
    The Jacobian of the PQC, at a point $\theta\in T^d$, is the differential $\mathrm{d}F$ expressed with respect to that basis,
    \begin{align}
        J_F(\theta)=\mathrm{d}F_\theta : \R^d =T_\theta T^d \to T_{F(\theta)}\groupG,
    \end{align}
    where $T_{F(\theta)}\groupG$ is the tangent space of $\groupG$ at the point $F(\theta)$, which is isomorphic to the Lie algebra $\mathfrak{g}$ of $\groupG$.
\end{definition}
$J_F(\theta)$ can be expressed conveniently as a matrix of shape $(d, 4^n)$.
From here on we will look at PQCs with $d=\dim(\groupG)\leq 4^n$.

Next, we prove that the image of a PQC does not have measure-zero holes poked into it.
While this is quite immediate from the nice properties that the involved spaces and the PQC have, it may provide some intuition as to how global surjectivity of PQCs can fail in the first place.

\densitysurjectivity*
\begin{proof}
    A PQC $F$ is a smooth, and in particular continuous, map and thus it maps its domain $T^d$, which is compact, to its compact image $F(T^d)\subset\groupG$. As $\groupG$ is a Lie group, it is a Hausdorff space, making $F(T^d)$ a closed set. 
    Now assume the image $F(T^d)$ to be dense in $\groupG$, so that its closure equals $\groupG$.
    As the image is already closed, it is the entire group $\groupG$.
\end{proof}

Next, let us prepare to prove \cref{lemma:full_rank_anywhere_ae}, which is forthcoming below.
For this, we will use that Pauli rotation gates together with Clifford gates can be traced to produce analytic functions, which come with strong properties. In particular, adjoining Pauli words with Clifford gates and Pauli rotation gates produces linear combinations of Pauli words with coefficients that are analytic functions in the circuit parameters.

\fullrankanywhereae*
\begin{proof}
    A convenient way to prove this lemma is to look at the (Jacobian) Gram matrix of the PQC $F$, constructed from its Jacobian matrix as 
    \begin{align}
    G_F: T^d &\to\R^{d\times d}\\
        \theta &\mapsto G_F(\theta)\\
        G_F(\theta)_{jk}
        &= \left\langle \frac{\partial F(\theta)}{\partial \theta_j}, \frac{\partial F(\theta)}{\partial \theta_k}\right\rangle
        =\frac{1}{2^n}\text{tr}\left(\frac{\partial F(\theta)^\dagger}{\partial \theta_j}\frac{\partial F(\theta)}{\partial \theta_k}\right),
    \end{align}
    where $\langle \cdot,\cdot \rangle$ denotes a trace inner product on $\C^{2^n\times 2^n}$.
    The Jacobian component $\tfrac{\partial F(\theta)}{\partial \theta_j}$ can be rewritten as 
    \begin{align*}
        \frac{\partial F(\theta)}{\partial \theta_j}&=i H_j(\theta) F(\theta)\\ 
        \text{ with }\ 
        iH_j(\theta) &= F_{[:j]}(\theta) \ iP_j\  F_{[:j]}(\theta)^\dagger.
    \end{align*}
    Here, $iH_j(\theta)$ is the so-called (left-)effective generator of $F$ for $\theta_j$, which is computed from the subcircuit $F_{[:j]}$ containing all gates up to the one using $\theta_j$, and from the (depending on convention, rescaled) Pauli word generator $iP_j$.
    $H_j(\theta)$ is Hermitian because $P_j$ is, and it can thus be decomposed in the Pauli basis. The coefficients in this basis are polynomials in trigonometric functions of the (potentially rescaled) parameter $\theta$, making them analytic functions.
    Thus, a Gram matrix entry $G_{jk}$, which equals the inner product of the coefficient vectors associated to the derivatives with respect to $\theta_j$ and $\theta_k$, is an analytic function as well.

    $J_F(\theta)$ has full rank if and only if $G_F(\theta)$ does, namely if and only if $\det{G_F(\theta)}\neq 0$.
    The determinant is an analytic function, as it is a polynomial of the Gram matrix entries.
    Following from the identity theorem, analytic functions can only take the value zero in line with two possibilities: they are the constant zero function, or they vanish on a set of measure zero within their domain. As we know that $J_F$, and thus $G_F$, has full rank at least at one point $\theta_r\in T^d$, we know that $\det{G_F(\theta_r)}\neq 0$ and thus that $\det{G_F}$ is not the constant zero function.
    It therefore only vanishes on a measure-zero subset of $T^d$, so that $J_F$ has full rank almost everywhere, i.e., the critical set $\mathcal{C}$ of $F$ has measure zero as well.
\end{proof}

\end{document}